%% file: main.tex
\documentclass[11pt]{article}
\usepackage{float}
\usepackage{amsmath, amssymb, amsthm}
\usepackage{graphicx}
\usepackage{latexsym}
\usepackage{amsfonts}
\usepackage{color}
\usepackage{dsfont}
\usepackage{comment}
\usepackage[left=2.5cm,top=3cm,right=2.5cm,bottom=3cm,bindingoffset=0.5cm]{geometry}
\usepackage[font=small,labelfont=bf]{caption}
\usepackage{hyperref}
\usepackage{xspace}

\newtheorem{theorem}{Theorem}
\newtheorem{lemma}{Lemma}
\newtheorem{proposition}{Proposition}

\newtheorem{corollary}{Corollary}
\newtheorem{remark}{Remark}

\newcommand\1{\mathds{1}}
\newcommand{\indi}[1]{\1_{\{#1\}}}

\graphicspath{{figs/}}

\def\gm{\gamma_{\theta,\rho}}

\def\am{\alpha_{\theta,\rho}}

\def\P{{\mathbb P}}
\def\R{{\mathbb R}}

\def\X{ {\cal X}}

\def\M{ {\mathcal M}}
\newcommand{\mc}[1]{\mathcal{#1}}
\newcommand{\mb}[1]{\mathbf{#1}}
\newcommand{\ms}[1]{\mathsf{#1}}
\newcommand{\sun}[1]{{#1}^{(n)}}

\newcommand{\pin}{\sun{\pi}}

\newcommand{\hPhi}{\widehat{\Phi}_{\theta}}

\def\hc{\hat{c}}
\def\hp{\hat{p}}

\def\ep{\hfill $\Box$}

\def\bp{\noindent{\bf Proof.}\ }
\def\epsilon{\varepsilon}

\newcommand{\mcn}[2]{\mathcal{#1}^{(n)}_{#2}}

\author{Matthieu Jonckheere \and Balakrishna J. Prabhu}

\newcommand{\Addresses}{{
	\bigskip
    \footnotesize
	\noindent
	 M.~Jonckheere, \textsc{Department of Mathematics, Universidad de Buenos Aires, 
			      Buenos Aires, Argentina}\par\nopagebreak
  \textit{E-mail address}: \texttt{matthieu.jonckheere@gmail.com}

				  \medskip
	\noindent
	B.~J.~Prabhu, \textsc{LAAS-CNRS, Universit\'e de Toulouse, CNRS, Toulouse, France}\par\nopagebreak
			  \textit{E-mail address}: \texttt{balakrishna.prabhu@laas.fr}
}}

\begin{document}

\title{Asymptotics of Insensitive Load Balancing and Blocking Phases\footnote{
This work was partially supported by the Basque Center for Applied Mathematics BCAM and the Bizkaia Talent and European Commission through COFUND programme, under the project titled "High-dimensional stochastic networks and particles systems", awarded in the 2014 Aid Programme with request reference number AYD-000-273, and by the STIC-AmSud 
project N\textsuperscript{o}\xspace 14STIC03.
}}

\maketitle

\abstract{
Load balancing with various types of load information has become a key component of modern communication and information systems. 
In many systems, characterizing precisely the blocking probability allows to
establish a performance trade-off between delay and losses.
We address here the problem of giving 
robust performance bounds based on the study of the asymptotic behavior of the insensitive load balancing schemes when the number of  servers and the load scales jointly. These schemes have the desirable property that the stationary distribution of the resulting stochastic network depends on the 
distribution of job sizes only through its mean. It was shown that they give good estimates of performance indicators for systems with finite buffers, generalizing henceforth Erlang's formula
whereas optimal policies are already theoretically and computationally out of reach for networks of moderate size.

We study a single class of traffic acting on a symmetric set of processor sharing queues with finite buffers and we consider the case where the load scales with the number of servers.
We characterize
  the response of symmetric systems under those schemes at different scales and show that three amplitudes of deviations can be identified according to whether $\rho < 1$, $\rho = 1$, and $\rho > 1$.
A central limit scaling takes place for a sub-critical load;
for $\rho=1$, the number of free servers scales like $n^{ {\theta \over \theta+1}}$
 ($\theta$ being the buffer depth and $n$ being the number of servers) and is of order 1 for super-critical loads.
This further implies the existence of different phases for the blocking probability. 
Before a (refined) critical load $\rho_c(n)=1-a n^{- {\theta \over \theta+1}}$, the blocking is exponentially small 
and becomes of order $ n^{- {\theta \over \theta+1}}$ at $\rho_c(n)$. This generalizes  the well-known Quality and Efficiency Driven (QED) regime or Halfin-Whitt regime for a one-dimensional queue, and leads to a 
generalized staffing rule for a given target blocking probability.

%

\vspace{4mm}

\input{01_introduction}
\input{02_review}
\input{03_model}
\input{04_meanfield}
\input{05_finerscales}

\input{06_numerical}

\input{07_practical}

\bibliographystyle{acm}
\bibliography{TheBib}
\Addresses

\appendix
\input{appendix}

\input{supplementary}

\end{document}

%% file: 01_introduction.tex

\section{Introduction}

Load balancing is a critical component in multi-servers systems such as call centers, server farms,
as well as in distributed systems with many applications running on different servers (as a single example, see the load balancing needs of the CERN network \cite{cern}).
Despite its intensive use, there are few efficient rules of thumb 
for dynamic load balancing schemes i.e. when the decision of the dispatcher depends on the instantaneous load (e.g. number of jobs) at each server,
 for which performance evaluation is within reach using currently known techniques. 
This fact becomes even more evident for large systems with asymmetric server speeds and blocking, where both the precise structure of optimal policies
(for specific traffic descriptions) and their performance elude current knowledge and techniques.

Two types of ideas were largely employed to overcome this difficulty: considering large scale networks 
and obtain asymptotic results using propagation of chaos on the one hand; restricting the load balancing
schemes to obtain more tractable reversible processes on the other hand.
We aim here at combining both techniques and show that it may lead to very precise results
 which are typically out of reach with other techniques. 
In turn, these results allow to give universal (i.e. valid for all job-size distributions) lower
bound of performance. In particular, different asymptotic scalings for the blocking probability can be identified: an exponentially small blocking probability for sub-critical loads; a polynomial order 
in the critical regime; and a constant level for super-critical loads.
Before describing our contribution more precisely, let us recall the vast effort of research in the two mentioned directions.

A first way of overcoming the computing difficulty of load balancing problems for processor sharing systems (and more generally symmetric queues) with generic job-size 
distribution is by restricting the 
routing policies so that the stationary regime of the system becomes invariant to the job-size distribution (except for its mean), leading to the insensitive load balancing (see \cite{bp}
for more details). To understand the underlying principles, it is useful to come back to the very properties of the Erlang formula which was clearly a revolution for performance evaluation of telephone networks and arguably the true start of queuing theory. The Erlang formula, which gives the probability of loss for a set of telephone lines, bases its lasting success on simplicity and robustness:
\begin{enumerate}
\item 
 the only assumptions which are required to apply the formula are Poisson arrivals and independent calls durations;
\item
the formula is insensitive to the call duration distribution and depends on a unique parameter:
the traffic intensity; and
\item
it can efficiently be computed using a recursive formula.
\end{enumerate}

At the mathematical level, the key property is the reversibility of the birth-and-death processes modeling the system under Markovian assumptions 
which implies the insensitivity property: the stationary measure of the system  does not depend on the whole call distribution but only on its mean.

\

Those same principles were translated to performance models of best effort and voice traffic in \cite{bp,BFreview} for models with balanced allocations 
and in \cite{bp, bonaldjonckheere04} for models with dynamic insensitive load balancing.  For multi-class networks with insensitive load balancing,
Markov Decision Programming techniques were employed in \cite{leinovirtamo06}, structural results were provided in \cite{jonckheeremairesse10}
 while extensive simulations were proposed in \cite{pla2008}.
  All this research dealt with networks of fixed size and allowed to draw the following conclusions.
 For networks with a unique class of traffic, the insensitive load balancing compares very accurately to optimal policies for a given 
 job-size distribution, while delay estimation are a bit less accurate \cite{bp,bonaldjonckheere04}. The penalization imposed by reversibility is greater for multi-class networks while the sensitivity (of optimal sensitive policies) also deteriorates \cite{leinovirtamo06,jonckheeremairesse10}.
 Hence a small to moderate price has to be paid for robustness and simplicity. It is perhaps counter-intuitive to notice that for models with infinite buffers, this price becomes very high. It was indeed proved that if the state space is infinite and in the absence of blocking,  the optimal insensitive load balancing (for any reasonable criterion)
is static (i.e., does not depend on the queue-lengths) and is hence much less efficient
than a state-dependent sensitive load balancing \cite{jonckheere06}. For sensitive schemes like join-the-shortest-queue, there is no characterization of the stationary measures in a general setting (see, for example, the approximations in \cite{KY11}).

On the other side of the spectrum, a specific attention has been given in the last decades to mean-field type results for different type of networks with load balancing applications. In particular, 
a great deal of research, which started with the seminal works of \cite{mitzenmacher, vvedenskaya96}, has been dedicated to prove mean-fields limits for schemes like 
join-the-shortest-of-$d$ (JSQ(d)) among $n$ queues where $n$ is large. For example, transient functional law of large numbers and propagation of chaos were obtained in \cite{graham00} for FIFO scheduling, \cite{stolyar2015}} obtains mean-field limit for the join-the-idle queue (JIQ) while \cite{mukherjee2015} computes the diffusive limit in the Halfin-Whitt regime for a class of policies of which JIQ and JSQ(d) policies are a special case.  For general service time distributions, propagation of chaos properties and asymptotic behaviour of the number of occupied servers were obtained for the JSQ  policy  in \cite{bramson2012}. All these results concern systems without blocking and are sensitive to the job-size distribution. For systems with blocking, the recent work of \cite{mukho15} computes the mean-field limit for the JSQ(d) scheme and exponentially distributed job sizes.

It is hence natural and complementary to look at insensitive networks with a very large number of servers 
and given buffer depth in order to see if the results obtained for finite networks scale appropriately.
As detailed below, this leads to very precise results which are qualitatively different from the case without blocking and are out of reach for sensitive policies with blocking. 
This in turn provides simple dimensioning rules.

\subsection*{Contributions}
We study the asymptotics of a set of $n$ processor sharing servers, each with buffer size $\theta$, 
fed by a Poisson process of intensity $\rho n$ when $n$ gets large, under the family 
of insensitive load balancing schemes shown to be optimal (in the class of insensitive load balancing) in \cite{bonaldjonckheere04}.
A more detailed description of the model is given in the next section. 

Building on closed-form expressions for the stationary measure, we characterize precisely 
the asymptotics of the stationary measure and the blocking probability for various scalings of the load. Consequently,  we provide universal benchmarks for achievable performance
which have no known counterpart for sensitive policies.

We first obtain the stationary measure of the number of occupied servers
and give its transient mean-field limit.
Considering the symmetric version of the model,
we show that the functional law of large numbers also holds for the stationary version of the system (limits in $n$ and $t$ commute). The existence and uniqueness of the limiting stationary probabilities are proved through a monotonicity argument involving the Erlang formula, while the stationary point is characterized through the Erlang formula.
 This implies simple conclusions on the asymptotic behavior of the blocking probability: the blocking is asymptotically vanishing for the sub-critical ($\rho < 1$) case and is equal to $1- \rho^{-1}$ for the super-critical case $\rho>1$. In both cases, this blocking probability corresponds to the optimal blocking probability achievable by any non anticipating policy.
Of course this is far from being sufficiently informative and
we are led to focus on a more detailed study of the stationary distribution for large $n$, establishing both
 large deviations principles for sub- and super-critical cases and moderate deviations results.
 We show that, when $\rho <1$ is fixed, the blocking probability is exponentially small, and we characterize the most probable deviations from the mean-field limit. The large deviation cost is shown to be a sum of two terms: the ``distance'' to the stationary point from distributions with a given mean plus the cost of having a different mean from the true stationary mean.
 We also show that a central limit theorem is valid for the occupation numbers around the stationary point of the mean-field in the sub-critical regime. 
 For the critical case $\rho=1$, the right scaling is not anymore of order $\sqrt{n}$.
  Using local limit theorems and exploiting the characterization of deviations from the mean-field limits, 
 we show that the number of free servers scales like $n^{ {\theta \over \theta+1}}$,
 the limiting distribution depending on $\theta$ and coinciding with the normal distribution only for $\theta=1$. 
 In a third step, we study the critical case at a finer scale and show that a 
qualitative phase transitions occurs at the critical load $\rho_c(n)=1 - a  n^{- {\theta \over \theta+1}}$ where $\theta$ is the buffer depth.
The blocking probability is exponentially small until $\rho_c(n)$ and of order $ n^{- {\theta \over \theta+1}}$ at this critical load.
This generalizes the Halfin-Whitt regime established for the $M / M / n  / n$ system,
(in that case, the correct scaling for the moderate deviations stayed of order $\sqrt{n}$),
and show that the popular staffing rule established for the  $M / M / n  / n$ system does actually change with the value of $\theta$
when load balancing is employed. The super-critical regime is simpler to characterize, the deviations being of order $1$. 
We illustrate these findings on simple numerical experiments. Finally, we comment on how these results can be used for performance planning, for instance
in trading delay for blocking, while controlling the level of blocking fixing the number of servers and how the level of information
needed for a possible implementation can be reduced. We also give insights on possible future work.


%% file: 02_review.tex
\section{Review of the optimal insensitive load balancing policy}
\label{sec:review}
{\em Notation.} 
We use the following notations common to all sections.
For any vector space (the exact one under consideration will be clear from the context), let $e_i$ be the point defined by
$(e_i)_i=1, \ (e_i)_j=0, j\neq i$.
We denote:
$$|x|= \sum_{i=1}^k x_i \quad \mbox{and}\quad {|x| \choose x} =
{|x|!\over x_1!\ldots  x_k!}.$$
We denote by ${\bf 1}_{S}$ the indicator
function of $S$, that is the map taking value 1 inside $S$ and 0
outside and denote respectively by $\R_+$ and $\R_+^*$ the set of
non-negative and positive reals.

This section is a review of the relevant definitions, merits and results known for the insenstive load-balancing policy investigated in this paper. The narrative here is for a more general model than the one we shall analyse. Nonetheless, it gives a flavour of the possible generalizations, some of which are elaborated upon in Section \ref{sec:eng}.

Consider a dispatcher and a set of $n$ processor sharing servers with speed $\mu_i$ for $i=1 \dots, n$.
Jobs with i.i.d. sizes sampled from a generic distribution of mean $1$ arrive to the dispatcher according 
to a Poisson process of intensity $\lambda$. The dispatcher routes an incoming job to one of the servers according to the following 
insensitive load balancing rule. 
Let $\theta=(\theta_1, \ldots,\theta_n)$
a vector of natural numbers and $\X \subset \mathbb N^n$ a finite coordinate convex set describing the constraints on the number of jobs in each server 
(for instance $\X= \{  x : x_i \le\theta_i, \forall i= 1 \ldots n \}$). Then, the incoming job is routed to server $i$ with probability
\begin{equation}
a^\theta_i(x)= {\theta_i - x_i \over \sum_{j=1}^N (\theta_j-x_j) } 1_{x +e_i \in \X}.
\end{equation}
Note that with this rule, the number of jobs in server $i$ is smaller than $\theta_i$ for all $i=1 \ldots n$. One can view $\theta_i$ as the buffer size 
of server $i$ but it could be a smaller number chosen to guarantee a certain rate of service.
Also note that this rule depends on the speeds $\mu_i$ only through the vector $\theta$.
Nevertheless this load balancing rule was proved to be optimal\footnote{Optimal in the sense that it minimizes the blocking probability or any convex criterion.} in the set of insensitive load balancing (for a unique class of traffic) in \cite{bonaldjonckheere04}, i.e., given the speeds $\mu_i$ there exists an optimal vector $\theta$ such that this rule is optimal among all insensitive load balancing.

Let $X$ be the stochastic process valued in $\X$ describing the number of on going jobs in each server.
Under Poisson arrivals and exponentially distributed job sizes, $X$ is a continuous-time
jump Markov process, on the state space $\X$, with infinitesimal
generator $Q=(q(x,y))_{x,y \in \X}$ given by
 $\forall x \in \X$,
\begin{equation}\label{eq-Q}
\begin{cases}
q(x,x-e_i)  \ = \  \mu_i & \mbox{if } x-e_i \in \X \\
q(x,x+e_i)  \ = \lambda  a_i(x) & \mbox{if } x+e_i \in \X\\
q(x,y)  \ = \ 0 & \mbox{if } y\in \X, \ y \neq x-e_i, x+e_i \:.
\end{cases}
\end{equation}

We recall that the family of insensitive load balancing corresponds to the routing rates $\lambda_i()=\lambda a(\cdot)$ such that
there exists a balance function $\Lambda: \X \rightarrow \R_+^*$, 
\begin{eqnarray}
\forall i,
\ \forall x \in \X, x+e_i \in \X, \quad \lambda_i(x) =\Lambda(x+e_i) /
  \Lambda(x), \label{eq:balance2}
\end{eqnarray}
which is equivalent to the detailed balance criterion. The relationship between this criterion and insensitivity was first formulated in \cite{whittle85}.
Under condition \eqref{eq:balance2}, the process $X$ is reversible and the stationary
distribution is given by
\begin{equation}\label{eq-reversible}
\pi(x) = {\Lambda(x)\Phi(x) \over \sum_{y\in \X} \Phi(y)\Lambda(y)},
\end{equation}
with $$\Phi(x)= \prod_{i=1}^n \mu_i^{-x_i}.$$
For the optimal insensitive load balancing corresponding to the mentionned rates $\lambda a_i(\cdot)$,
 the routing balance function $\Lambda$ takes the form
\begin{equation}
	\Lambda(x)=\Lambda_\theta(x)={|\theta-x| \choose \theta -x} \lambda^{|x|},
\end{equation}
where ${|\theta-x| \choose \theta -x}= {|\theta-x|! \over \prod_{i=1}^n (\theta_i-x_i) !}$ are the multinomial coefficients. 


The blocking probability, $B_{\theta}$, of an arriving job can be determined using the PASTA property to be  $\pi(\theta)$.

%% file: 03_model.tex
\section{Model and preliminary results}
\label{sec:model}

In the rest of the paper, we shall assume that the servers are homogeneous, that is, they have
the same speed and the same buffer size. (We shall comment later on the possibility of extending those results).
Without loss of generality, let the common speed be $1$. The common buffer
size will be taken to be $\theta$. (From now on, $\theta$ is a natural number and not a vector as in the previous section).
For the asymptotic analysis we have in mind, it turns out to be more convenient to define the state as the number 
of servers processing a certain number of jobs instead of the number of jobs being processed in every server.  
Let $\mc{S} = \{s \in \{0,1,\hdots, n\}^{\theta+1} : \sum_{i=0}^{\theta} s_i = n\}$ be the set of states where $s_i$ 
corresponds to the number of servers with $i$ jobs. In state $s\in\mc{S}$, the insensitive load balancing rule described in Section 
\ref{sec:review} will route an incoming job to a server with $i$ jobs at rate
\begin{equation}
	\lambda_i(s) = \lambda\frac{(\theta - i)s_i}{n\theta-\bar{s}},
\label{eqn:arr_ho}
\end{equation}
where $\bar{s} = \sum_{i=0}^\theta i s_i$.

Let $\{\sun{S}(t) \in \mc{S}\}_{t\geq 0}$ be a stochatic process denoting, at time $t$, the number of servers with 
$i$ jobs, $i = 0,\ldots, \theta$. Under Poisson arrivals and exponentially distributed job sizes, $\sun{S}(t)$ is a 
continuous-time jump Markov process on the state space $\mc{S}$ with the following transition rates
\begin{equation}
\sun{S}(t) \to
\begin{cases}
		 \sun{S}(t) + e_i - e_{i-1}	& \mbox{at rate } \lambda_{i-1}(s), i\geq 1; \\
		 \sun{S}(t) + e_i - e_{i+1}	& \mbox{at rate } s_{i+1}, \\
\end{cases}
\label{eqn:trans_ho}
\end{equation}
assuming that the transitions take the process to a state within $\mc{S}$.

The reversibility property of $X$ is preserved by $S$. More precisely,
\begin{theorem}
If the job-size distribution is exponential, the process $\sun{S}(t)$ is a reversible Markov process and its stationary distribution is given by
	\begin{align}
			\sun{\pi}(s) = \sun{\pi}_0\frac{(n\theta - \bar{s})!}{(n\theta)!}\binom{n}{s}\prod_{k=0}^{\theta}\left(\frac{\theta!}{(\theta-k)!}(n\rho)^k\right)^{s_k}, 
	\label{eqn:pis_ho}
\end{align}
where is the total number of jobs in the system, and $\rho = \lambda/n$ is the load per server, and $\sun{\pi}_0$  corresponds
to the probability of the state with all servers empty, that is, $\bar{s}=0$ and $s = (n,0,\hdots,0)$.
\end{theorem}

\begin{proof}
		A sufficient condition for a probability measure to be the stationary measure of a Markov chain is that it satisfy the local balance equations.
		Consider two states $s$ and $s + e_i - e_{i-1}$, both within $\mc{S}$. From \eqref{eqn:pis_ho},
	\begin{align}
			\frac{\pin(s + e_i - e_{i-1})}{\pin(s)} &= \frac{\lambda (\theta-(i-1)) s_{i-1}}{n\theta-\bar{s}}\frac{1}{(s_i+1)\mu}, \\
						&=\frac{\lambda_{i-1}(s)}{(s_i+1)},
	\end{align}
	which are in the same proportion as the local transition rates between these two states as computed from \eqref{eqn:trans_ho}.
\end{proof}
\begin{corollary} 
		Using the PASTA property, the blocking probability is given by
		\begin{equation}
				\sun{B}_\theta = \sun{\pi}_0\frac{(n\rho)^{n\theta}(\theta!)^n}{(n\theta)!}.
		\label{eqn:bth_n}
		\end{equation}
\end{corollary}
Instead of using $\pin_0$ as the normalizing constant, we can resort to $\sun{B}$ for this purpose, and rewrite \eqref{eqn:pis_ho} as
\begin{equation}
		\sun{\pi}(s) = \sun{B}_\theta (n\theta - \bar{s})!\binom{n}{s}\prod_{k=0}^{\theta}\left(\frac{1}{(\theta-k)!}(n\rho)^{k - \theta}\right)^{s_k},
\label{eqn:pis_btho}
\end{equation}

A special case that will reappear throughout this paper is the one with $\theta=1$, which corresponds to the classical $M/M/n/n$ queue or the Erlang loss system.  
Upon setting $\theta=1$ in \eqref{eqn:pis_ho}, we obtain:
\begin{align}
		\pin(s_0) &= \frac{(n\rho)^{(n-s_0)}}{(n-s_0)!}\pi(0),\label{eqn:pi_theta_1}\\
				\mbox{where } \pin_0 &= \sum_{k \leq n} \frac{(n\rho)^{n-k}}{(n-k)!}.
\end{align}
where $s_0$ is the number of empty servers. We hence retrieve the formula corresponding to the $M/M/n/n$ queue, as expected.

We remind the reader that for the JSQ policy, the stationary measure is quite intricate to compute even for the case of two servers \cite{KY11}, 
making it difficult to predict the performance of this policy.

Once the stationary measure is determined, the stationary performance measures such as the mean soujourn time and the blocking probability can
be numerically computed for any given set of parameters such as $n$, $\theta$, or $\rho$. However, for large $n$ or $\rho$ close to $1$, 
the relationship between the performance measures and the parameters can be obtained in a more palatable (and exploitable) form  using asymptotic analysis.
The following sections will follow this path leading to a mean-field limit as well as the characterization of the large and moderate deviations.

%% file: 04_meanfield.tex
\section{A mean-field deterministic limit}
In this section, we give the limiting transient and stationary behaviour in case of exponentially distributed job-sizes
when $n$ diverges. This limit called the mean-field limit has become a classical asymptotic regime for the analysis of large queuing systems and particle systems 
with a large number of servers or particles \cite{LeBoudec13,benaim99, kipnis2013,mukho15}. In load-balancing applications, this type of analysis has 
been used for several policies whose stationary measure is either unknown or known for relatively small values of the number of servers (for example, 
shorter of $d$ choices \cite{mitzenmacher}, joining the shortest queue \cite{mukho15}).  Indeed, for data-centers that can have hundreds to thousands of servers, the mean-field 
limit can give a first-order approximation to the system behaviour both in the transient and in the stationary phase.

In the mean-field limit, dynamics for the fraction of servers containing a certain number of jobs are as follows.
 \begin{theorem}
		 Fix a $\rho$ and a $\theta \geq 1$. For exponentially distributed job-sizes, for all fixed time, $\sun{S}(t)/n \xrightarrow{{\cal L}^2} y(t)$, which is the solution of 
		 the following set of differential equations:
\begin{align}
		\frac{dy_j(t)}{dt} &= \rho\frac{\theta - (j-1)}{\theta - \sum_k k y_k(t)}y_{j-1}(t) +  y_{j+1}(t)\label{eqn:mf_1}\\
		&  - \rho\frac{\theta - j}{\theta - \sum_k k y_k(t)}y_j(t) - y_j(t) , \; 0 < j < \theta, \\
		\frac{dy_\theta(t)}{dt} &= \rho\frac{1}{\theta - \sum_k k y_k(t)}y_{\theta-1}(t)  -  y_\theta(t), \\
		\frac{dy_0(t)}{dt} &=  y_{1}(t) -\rho\frac{\theta}{\theta - \sum_k k y_k(t)}y_0(t)\label{eqn:mf_2}. 
\end{align}
with $y(0) = \lim_{n\to\infty} \frac{\sun{S}(0)}{n}$.
\end{theorem}
\bp
For Poisson arrivals and exponentially distributed job-sizes, the technical difficulty is much less compared to that for generic job-size distribution.
 We hence do not give the details of the proof but only sketch the argument. Fix a time interval $[0,t]$. The process is trivially tight and assuming exponential 
job-size makes it Markov. Dynkins formula allows to write this Markov process as a drift part plus a martingale and calculating the increasing process 
of the martingale, one proves that the martingale part goes to $0$. As a consequence, the process converges along subsequences to a deterministic process in 
${\cal L}^2$. It remains to prove that the limit is unique which is easy in this case since using the regularity of the rates of the process, the limit is 
a differential equation with a Lipschitz drift. 
\ep

\begin{remark}
We expect the results to hold under generic job-sizes distribution but the proof becomes much more technical
as one has to work with measure-valued processes and falls out of the scope of this paper.
Results like asymptotic independence for randomized load balancing schemes such as join the shortest of d queues with generic job-sizes have been proved in \cite{bramson2012}.
\end{remark}

\begin{theorem}
		For $0 < \rho \leq 1$, the unique steady-state solution of the system of equations \eqref{eqn:mf_1}--\eqref{eqn:mf_2} is given by
		\begin{align}
				\hp_j & = \left(\frac{\theta-\hc}{\rho}\right)^{\theta-j} \frac{1}{(\theta-j)!}\hp_\theta, \label{eqn:pj}\\
				\mbox{with  } \hp_\theta &= \frac{1}{\sum_{k=0}^\theta \left(\frac{\theta-\hc}{\rho}\right)^k \frac{1}{k!}}. \label{eqn:p0}
		\end{align}
		where 
		\begin{equation}
				\hc=\theta- \rho Erl^{-1}_{\theta}(1-\rho),
		\label{eqn:erl_c}
		\end{equation}
		with $Erl_\theta^{-1}$ as the inverse function of the Erlang blocking viewed as a function of the traffic intensity for a fixed buffer depth $\theta$.

		If $\rho>1$, the unique solution is $\hc= \theta$, $\hat p_j=0$, for $j \le \theta-1$ and $\hat p_{\theta}=1$ .
\label{the:mf_st}
\end{theorem}

\begin{proof}
		Suppose first $\rho <1$. It can be easily verified that, with 
		\begin{equation}
				\hc = \sum_{k = 0}^\theta k \hp_k,
		\label{eqn:mean_number}
		\end{equation}
		\eqref{eqn:pj} and \eqref{eqn:p0} is the steady-state solution of \eqref{eqn:mf_1}--\eqref{eqn:mf_2}. We now show that
		$\hc$ as defined in \eqref{eqn:mean_number} verfies \eqref{eqn:erl_c}. 
		
		 After some simple algebraic manipulations, it can be verified that the fixed point equation \eqref{eqn:mean_number} is equivalent to the equation
		\begin{equation}
		(1-\rho)\sum_{k=0}^\theta \left(\frac{\theta-x}{\rho}\right)^k\frac{1}{k!} = \left(\frac{\theta-x}{\rho}\right)^\theta\frac{1}{\theta!}
		\end{equation} in the set $[0,\theta]$. 
		Thus solving \eqref{eqn:mean_number}  boils down to finding a traffic intensity $a:=\frac{\theta-x}{\rho} $ such that the Erlang blocking formula with intensity 
		$a$ gives $1-\rho$, i.e.:
		$$
			Erl_\theta(a)= {a^\theta\frac{1}{\theta!} \over \sum_{k=0}^\theta a^k\frac{1}{k!}}= (1-\rho).
		$$
		By a simple sample path argument, the Erlang formula is an increasing function of $a$ that is $0$ in $0$ and $1$ in $+\infty$.
		Hence, it is invertible and there is a unique $a>0$ such that $Erl_\theta(a)=1-\rho$. Now observe that by a conservation argument (the traffic entering vs traffic outgoing) we have that for all $a$:
		$$ a(1- Erl_\theta(a)) \le \theta,$$
which boils down (given the definition of $a$) to 
$$ a \le {\theta \over \rho}. $$
which in turns gives a unique  solution to
$${\theta-x \over \rho}=a.$$

Now consider $\rho > 1$. It is straightforward to verify that $\hc= \theta$, $\hat p_j=0$, for $j \le \theta-1$ and $\hat p_{\theta}=1$ is a solution. 
  If $\hat c < \theta$, then the drift of the differential equation (19) cannot be $0$
which implies that the solution given is unique.

\end{proof}

Using the generic method developed in \cite{LeBoudec13} to inverse limits  in $n$ and $t$  under the assumption of reversibility of the Markov process under study (which is indeed verified here), we can state that
\begin{proposition}
		For $\rho<1$, $\sun{\pi}$ converges point wise to $p$ when $n$ and $t$ converge to infinity.
\end{proposition}
\bp
 \cite{LeBoudec13} allows to state that if there is convergence for fixed time intervals, if the process is reversible for fixed $n$, and if there exists a unique stationary limiting point, (which we have verified) in the previous Proposition, then the conclusion of the Proposition hold.
\ep

\begin{remark}
By insensitivity, taking limit in time first define a sequence of limiting distribution $\sun{\pi}$ which do not depend on the specific job-size distribution and which converges towards  $p$.
\end{remark}

\subsection{Performance consequences}
Let $B_{\theta}$ denote the stationary blocking probability in the mean-field limit, that is, when $n\to \infty$ and $t \to \infty$. Using the PASTA property for fixed $n$, the blocking probability
of a job is the probability that it finds all the servers in their blocking state $i.e.,$ that upon arrival all the servers have $\theta$ tasks.

Before deriving $B_\theta$, we first give a lower bound on the blocking probability that could be achieved by any non-anticipating and size-unaware load balancing policy.
	\begin{proposition}
			For $\theta > 0$, the blocking probability of any non-anticipating and size-unaware load balancing policy is greater than $\max(0,1 - \rho^{-1})$.
	\label{prop:lower}
	\end{proposition}
	\begin{proof}
		We give the argument for $\rho > 1$. The argument for $\rho \leq 1$ is similar. Consider the system in which the resources are pooled, that is, 
		there is one server of service rate of $n\mu$ and buffer size $n\theta$. By a path-wise argument for Markovian versions of the systems (implying by insensitivity 
		the result for all service times in stationary regime), the blocking probability of this system will be less than any system with a set of disjoints servers.
		In the pooled system, the scaled number of tasks in the system $\frac{X(t)}{n}$ will follow the differential equation:
		\begin{align}
				\dot{x}(t) &= \rho - 1, \; 0 < x(t) < \theta, \\
				\rho(1 - \hat{B}_\theta) &= 1, \; x(t) = \theta.	
		\end{align}
		When $x(t)$ is in the interior of the state space, all tasks will be accepted. On the boundary $x(t) = \theta$, the tasks which cause overflow will be blocked. Hence, the
		blocking probability  will be 
		$$
			\underbar{B}_\theta = 1 - \rho^{-1}.
		$$
	\end{proof}
	 We now be shown that the insensitive load-balancing policy achieves this lower bound which is independent of $\theta$.
	\begin{proposition}
		The limiting blocking probability of the insensitive load balancing policy is given by
		\begin{equation}
				B_{\theta} = 
				\begin{cases}
						0 & \mbox{if } \rho < 1;\\
						1 - \rho^{-1} & \mbox {otherwise}.
				\end{cases}
		\end{equation}
	\end{proposition}
		\begin{proof}
		For $\rho < 1$, it can be seen that $\hp_\theta < 1$, and therefore, $B_{\theta} = 0$.	
	
		For $\rho \geq 1$, we shall prove first this result for $\theta = 1$. 
		For $\theta = 1$, from \eqref{eqn:mean_number},
		$$
			c = p_1,  
		$$
		and from \eqref{eqn:pj},
		$$
			p_1 = \frac{\rho(1-B_1)}{1 - p_1}(1-p_1) = \rho(1 - B_1).
		$$
		Since $p_1 \leq 1$, $B_1$ is such that $\rho(1-B_1) = 1$. That is, $B_1 = 1 - \rho^{-1}$.

		Using Proposition 3 in \cite{bonaldjonckheere04},  $B_{\theta} \le B_1$ for all $\theta \ge 1$.
		Now, using the lower bound in proposition \ref{prop:lower}, this implies that $B_{\theta}=B_1$.
	\end{proof}

	The stationary blocking probability of the insensitive policy is thus minimal in the considered class of policies, and it is independent of $\theta$. Hence, 
   	even a buffer of size $1$ is sufficient to get the optimal stationary behavior. We would like to point out that the optimality is only valid
	in the limit $n\to\infty$. In order to compute the blocking probability (or other performance measures) for values of $n$ that are large but finite, 
	one has to look at finer scales, which will be the objective of Section \ref{sec:finer}.

%% file: 05_finerscales.tex
\section{Finer scales and estimates}
\label{sec:finer}
While the results of the previous section are interesting for some performance metrics like the mean number of customers 
(which gives the mean waiting time via Little's formula), they are too rough to be really informative
in terms of blocking probabilities. Indeed, any reasonable dynamic load balancing may achieve the given bounds.
To get useful and discriminative estimates, we hence need to investigate the process $\sun{S}$ at finer scales.
In particular, we aim at determining when blocking can be considered a large deviation event
(with a probability exponentially small in $n$) and when it will be in the scale of the central limit theorem.

\subsection{Large  deviations}
Let $\mc{P} = \{p \in \R_+^\theta : \sum_{i=0}^{\theta} p_i = 1\}$.
For $c > 0$, denote $\mcn{S}{c} = \{s \in \mathcal{S} : \bar{s} = nc\}$, 
and $\mcn{P}{c} = \{q \in \mc{P} : nq \in \mcn{S}{c}\}$. Since $\bar{s} = \sum_k k s_k$, we have $\sum_k q_k = c$, $\forall q \in \mcn{S}{c}$.
Thus, $\mcn{P}{c}$ is the set of discrete probability distributions taking values on a lattice of unit size $1/n$ and having a first moment of $c$.

Define $p \in \sun{\mc{P}}_c$ by
\begin{equation}
		p_k(c) := \frac{1}{(\theta-k)!}\left(\frac{\theta - c}{\rho}\right)^{\theta - k}\frac{1}{\psi(c)}.
\label{eqn:pkc}
\end{equation}
where 
\begin{align}
		\psi(c) &=  \sum_{k=0}^\theta \frac{1}{k!}\left(\frac{\theta - c}{\rho}\right)^{k},
\label{eqn:psi}
\end{align}
is a normalizing constant which ensures that $p$ is a probability vector. There need not be a vector in $\mcn{P}{c}$ satisfying \eqref{eqn:pkc}, in which
case we define $p$ to be the vector\footnote{When the value of $c$ is clear from the context, we will use the notation $p$ instead of $p(c)$.} 
in $\mcn{P}{c}$ which is closest (say in norm $l^1$) to satisfying \eqref{eqn:pkc}. 
To simplify the notation, let $\sun{\pi}(q;c) = \pin(nq)\indi{q\in\mathcal{P}_c}$ be the stationary probability of observing  $q \in \mcn{P}{c}$. 

Let 
\begin{equation}
		\bar{c} = \arg\max_{c \in [0,\theta]} e^{c}\psi(c).
\label{eqn:copt}
\end{equation}

We first characterize $\bar c$. As a consequence of the definition of $\psi$, simple algebric computations show that  $\bar c$ coincides (as it intuitively should) with 
the asymptotic mean value $\hat c$ found in Theorem \ref{the:mf_st}, i.e.,
\begin{proposition}
\label{pro:bc}
$\bar{c}=\hat c= \theta- \rho Erl^{-1}_{\theta}(1-\rho)$. 
\end{proposition}
The proof of the proposition is similar to the one for $\hc$ in Theorem \ref{the:mf_st}.

Let $\hp = p(\hc)$.  The large deviation cost is shown to be a sum of two terms: the ``distance'' to the stationary point from distributions with equal means plus the cost of having a different mean from the stationary mean.
More precisely, we have the following  large deviation estimates for $\sun{S}$:
\begin{theorem}
For  $\rho < 1$,
\begin{equation}
		\begin{split}
		\lim_{n\to\infty} \frac{1}{n}\log\left(\frac{\pin(q;c)}{\pin(\hp;\hc)}\right) = (c - \hat{c}) + \log\left(\frac{\psi(c)}{\psi(\hat{c})}\right) \\
				- D_{KL}(q(c)\Vert p(c)).
\end{split}
\end{equation}

\label{thm:ldpc}
\end{theorem}

\begin{proof}
Applying Stirling's approximation in the term containing $n\theta$ in \eqref{eqn:pis_btho} and noting that $\sum_k k q_k = c$, we get
\begin{align} 
		\pin(q;c) &\sim \sun{B}_\theta (2\pi(n\theta -nc))^{1/2} e^{-n\theta + nc} \nonumber \\
		&\quad\cdot \binom{n}{nq}\prod_{k=0}^{\theta}\left(\frac{1}{(\theta-k)!}\left(\frac{n\theta - nc}{n\rho}\right)^{\theta-k}\right)^{nq_k} \\
		&= \sun{B}_\theta (2\pi(n\theta -nc))^{1/2} e^{-n\theta + nc} \psi^{n} 
		\cdot\binom{n}{nq}\prod_{k=0}^{\theta}p_k^{nq_k},
\label{eqn:asym_pi}
\end{align}

Thus, $\pin(q;c)$ is proportional to the multinomial distribution with $p_k$ as the probability of success of the $k$th class. 
Using Stirling's approximation in \eqref{eqn:asym_pi}, we get
\begin{align}
		\pin(q;c) &\sim \sun{B}_\theta (2\pi(n\theta -nc))^{1/2} e^{-n\theta + nc} \psi^{n} (2\pi n)^{-\theta/2} \nonumber \\
		&\quad\cdot \prod_{k=0}^{\theta}\left(\frac{p_k}{q_k}\right)^{n q_k}\frac{1}{q_k^{1/2}},
\label{eqn:asym_piq}
\end{align}
from which the desired result can be deduced.
\end{proof}

\begin{corollary}
For $\rho < 1$,
\begin{equation}
\lim_{n\to\infty} -\frac{1}{n}\log\left(\frac{\pin(q;c)}{\pin(p;c)}\right) = D_{KL}(q\Vert p).
\end{equation}
\label{cor:ldpc}
\end{corollary}
Corollary \ref{cor:ldpc} says that, conditioned on observing $nc$ jobs in the system, the probability of observing a certain distribution of jobs over the servers concentrates 
around $p$. The probability of observing any other $q \in \mcn{P}{c}$ decreases exponentially with rate $nD_{KL}(q\Vert p)$, that is the  Kullback-Liebler distance  
serves as the large-deviations rate function. This result is akin to Sanov's theorem in information theory \cite{coverthomas2006}. 

\begin{corollary}
For $\rho < 1$,
\begin{equation}
		\lim_{n\to\infty} \frac{1}{n}\log\left(\frac{\pin(p;c)}{\pin(p;\hat{c})}\right) = (c - \hat{c}) + \log\left(\frac{\psi(c)}{\psi(\hat{c})}\right).
\end{equation}
\label{cor:ldpp}
\end{corollary}
Corollary \ref{cor:ldpp} states that the scaled number of tasks in the system concentrates around $\hat{c}$ exponentially with rate$(c-\hat{c}) + \log(\psi(c)/\psi(\hat{c}))$.

\subsection{Asymptotics for the blocking probability}
In this section, we shall look at the asymptotics for the blocking probability which is the main performance measure of interest 
for our system. Two different asymptotic regimes will be considered: {\it (i)} the number of servers, $n$ scales, linearly with the 
arrival rate, $n\rho$; and {\it (ii)} the Halfin-Whitt regime \cite{HW81} which has a linear term as in {\it (i)} along with a sub-linear 
term that represents the safety margin. 

The starting point for both these asymptotic regimes will be an integral characterization of $B_{\theta}$ which is derived from the generating function of the 
stationary measure (see Theorem \ref{thm:gf} in the supplementary material). 

\begin{theorem}
		For $\rho \in (0,1)$, the blocking $B_{\theta}$ has the asymptotic form: 
		\begin{equation}
				\lim_{n\to\infty}\sun{B}_{\theta} \exp(nR(\gm))\left(\frac{2\pi n}{\alpha_{\theta,\rho}}\right)^{1/2} = 1.
				\label{eqn:bt_asy}
		\end{equation}
		where
\begin{equation}
		R(t) = \log\left(\sum_{k=0}^{\theta}\frac{t^k}{k!} \right) - \rho t,
\end{equation}
		  
\begin{align}
	\gm = \arg\max_{t\in (0,\infty)} R(t) =  \frac{\theta-\hat{c}}{\rho},
\end{align}
		and
		\begin{equation}
				\alpha_{\theta,\rho} =\frac{(1-\rho)}{\rho}\left(\frac{\theta}{\rho\gm} - 1\right).
		\label{eqn:alpha}
		\end{equation}
\label{thm:btheta}
\end{theorem}
\begin{corollary}
For $\theta = 1$, $\gm= \frac{1-\rho}\rho^{-1}$ and $\alpha_{\theta,\rho} = 1$. Thus,
\begin{equation}
		\sun{B}_1 \sim e^{n(1-\rho)}\rho^n(2\pi n)^{-1/2}.
\end{equation}
\end{corollary}
For the proof of Theorem \ref{thm:btheta}, we shall need the following result whose proof follows along the same lines as that of 
Theorem \ref{the:mf_st}.
\begin{lemma}
		Let $\gm = \arg\max_{t\in (0,\infty)} R(t)$. Then, $\gm$ is the unique solution of the equation
		\begin{equation}
				(1-\rho)\sum_{k=0}^{\theta}\frac{x^k}{k!} = \frac{x^\theta}{\theta !},
				\label{eqn:gamm}
		\end{equation}
\label{lem:gm}
is $ x= Erl^{-1}_\theta(1-\rho).$
\label{lem:R}
\end{lemma}

\begin{proof}[Proof of Theorem \ref{thm:btheta}]
Using calculations on the generating functions for the stationary probability (see the supplementary material), we obtain that
\begin{align}
		\lim_{n\to\infty} \sun{B}_{\theta} n\rho \int_0^\infty \left(\sum_{k=0}^{\theta}\frac{1}{k!} t^k \right)^{n}e^{-t n\rho} dt = 1.
\label{eqn:bt_int}
\end{align}
The asymptotic form of the integral can be determined using Laplace's method, which says
that
\begin{equation}
		\int_0^\infty e^{nf(t)} dt \approx e^{nf(t_0)}\left(\frac{2\pi}{n(-f^{\prime\prime}(t_0))}\right)^{1/2},
\label{eqn:lap}
\end{equation}
where $t_0$ is the unique maximizer of $f$ in $(0,\infty)$.

Define:
\begin{equation}
		R(t) = \log\left(\sum_{k=0}^{\theta}\frac{t^k}{k!} \right) - \rho t.
\end{equation}

Then
\begin{align}
		\lim_{n\to\infty} \sun{B}_{\theta} e^{n R(t_0)}\left(\frac{2\pi n\rho^2}{-R^{\prime\prime}(t_0)}\right)^{1/2}= 1.
\label{eqn:bt_lap}
\end{align}
where $t_0$ maximizes $R$ and is characterized in Lemma \ref{lem:R}. For Laplace's method to be applicable, one needs $R^{\prime\prime}(t_0) < 0$. 
This is shown in Lemma \ref{lem:Rconc} which appears in Appendix \ref{ssec:Rconc}.
		We shall now compute $\alpha_{\theta,\rho} = -\rho^{-2} R^{\prime\prime}(\gm)$, and the main result then follows from
		\eqref{eqn:bt_lap}.
		
		Let $g_\theta(t) = \sum_{k=0}^{\theta}\frac{t^k}{k!}$.
		Since $\gm$ is the maximizer of $R(t)$, we have $R^\prime(\gm) = 0$, which upon rearrangement gives
			\begin{equation}
					\frac{g_{\theta-1}(\gm)}{g_\theta(\gm)} = \rho.
			\end{equation}
			From the definition of $g_\theta$ and Lemma \ref{lem:gm}, we have
$				g_\theta(\gm) = \frac{1}{1-\rho}\frac{\gm^\theta}{\theta!},$
		and
		\begin{align}
				g_{\theta-2}(\gm) &= g_{\theta-1}(\gm) - \frac{\gm^{\theta-1}}{(\theta-1)!}  \\
				&= \frac{\gm^\theta}{\theta!}\left(\frac{\rho}{1-\rho} - \frac{\theta}{\gm}\right)
		\end{align}
		Thus,
		\begin{align}
				R^{\prime\prime}(\gm) &=  \frac{g_{\theta-2}(\gm)}{g_\theta(\gm)} - \left(\frac{g_{\theta-1}(\gm)}{g_{\theta}(\gm)}\right)^2\\
				&=\rho - \frac{(1-\rho)\theta}{\gm} - \rho^2 
		\end{align}
		from which the expression for $\am$ follows. 
\end{proof}
%

For $\rho > 1$, we cannot use the directly use the technique that was used for $\rho < 1$ because the maximum of $R(t)$ in the 
interval $[0,\infty)$ occurs at $t=0$, which is not an interior point of the support of $R$. So, we shall resort to a theorem
due to Erdelyi that treats this case. 

\begin{theorem}
		Let $\rho > 1$ and $n(\rho-1)$ be bounded away from $0$. As $n \to \infty$,
\begin{equation}
		\sun{B}_{\theta} \sim 1-\rho^{-1} +\frac{1}{(\rho-1)^{\theta}n^\theta} + o(n^{-\theta}).
\end{equation}
\end{theorem}
\begin{proof}
		We shall apply the Erdelyi theorem (see Theorem $1.1$ in the arXiv preprint of \cite{Nemes2012} for a precise statement 
		with the notation relevant for our proof) with 
		$f(t) = -R(t)$ and $[a,b) = [0,\infty)$. Let us verify that the four conditions of this theorem are satisfied by $-R(t)$.
	
	For $\rho \geq 1$, the function $-R(t)$ is increasing in the interval $[0,\infty)$ with minimum at $t=0$. To see this,
	\begin{align*}
		-R^\prime(t) &= \rho - \frac{\sum_{k=0}^{\theta-1}\frac{t^{k}}{k!}}{\sum_{k=0}^{\theta}\frac{t^k}{k!}} 
				  \geq 1 - \frac{\sum_{k=0}^{\theta-1}\frac{t^{k}}{k!}}{\sum_{k=0}^{\theta}\frac{t^k}{k!}} \geq 0.
	\end{align*}
	From Lemma \ref{lem:asym_h}, in a neighborhood of $0$,  $-R(t)$ is analytic with expansion 
	\begin{equation}
			-R(t) = (\rho - 1)t + \frac{t^{\theta+1}}{(\theta+1)!} + o(t^{\theta+1}),
	\end{equation}
	and $-R(t)$ is continuously differentiable with an analytic derivative. Finally, to show the absolute convergence of the integral,
	note that
	\begin{align}
		\int_{0}^\infty e^{-n(-R(t))}dt \leq \int_{0}^\infty e^{-n(\rho-1)t}dt = \frac{1}{n(\rho-1)} < \infty,
	\end{align}
	as long as $n(\rho-1)$ is bounded away from $0$. Thus, all the necessary conditions required by Erdelyi theorem are satisfied. 
	
	The various parameters in the asymptotic expansion $(1.5)$ in \cite{Nemes2012} are: $R(0) = 0$, $\alpha = 1$, $a_0$ = $\rho - 1$, 
	$a_1 = \hdots = a_{\theta-1} = 0$, $a_{\theta} = \frac{1}{(\theta+1)!}$, which gives $\beta_0 = (\rho - 1)^{-1}$, $\beta_1 = \hdots \beta_{\theta-1} = 0$ 
	and $\beta_\theta = -\frac{(\rho-1)^{-(\theta+2)}}{\theta!}$, so that
	\begin{equation}
			\int_{0}^{\infty}e^{nR(t)}dt \sim \frac{1}{(\rho-1)n} - \frac{1}{(\rho-1)^{\theta+2} n^{\theta+1}} + o(n^{-(\theta+1)})
	\end{equation}
	Substituting the above asymptotic expansion in \eqref{eqn:bt_int}, we get the claimed result.
\end{proof}
Since $R(0) = 0$, there is no exponential decay of the blocking probability when $\rho > 1$.
\begin{corollary}
		Setting $\theta = 1$ in the above theorem, we get the corresponding result for $\theta = 1$ obtained in \cite{Jagerman74} (see Theorem $13$ in there). 
\end{corollary}

The previous theorems give the asymptotics of the blocking probability for a fixed load per server for large number of servers. For $\rho < 1$, the blocking probability 
goes to $0$ exponentially quickly in $n$ while for $\rho > 1$ it goes to $1-\rho^{-1}$, a strictly positive quantity. The next theorem looks at the scaling law that results in a polynomial
blocking probability.
For the Erlang C model, that is, a system without blocking, this regime has the following interpretation:  
if the cost of servers is high, Halfin and Whitt \cite{HW81} observed that it could be beneficial to reduce the number of servers in such a way such that the probability
 of waiting is no longer exponentially small but decays as $n^{-1/2}$. This increase in the waiting probability has the benefit of requiring $\lambda + O(\lambda^{1/2})$ instead 
 of $b\lambda$, $b>1$, servers. Thus, one gains in the cost and the utilization of servers at the expense of the waiting probability. In order 
to evoke this trade-off between these two quantities, this scaling regime is also called the Quality and Efficiency Driven (QED) regime. We note that this 
asymptotic regime was already studied for the Erlang B system in Jagerman \cite{Jagerman74} (see Theorem $14$) 
 but the interpretation 
in terms of a trade-off is due to Halfin and Whitt \cite{HW81} for systems without blocking and Whitt \cite{Whitt84} for systems with blocking.

The following theorem gives the QED scaling for the balanced load-balancing policy and can be viewed as a generalization of the QED result for the Erlang loss
model. 
\begin{theorem}
		For $a\in(\infty,\infty)$, let 
		\begin{equation}
				n\rho = n + a n^{1/(\theta+1)}.
		\label{eqn:hw_th}
		\end{equation}
		Then,
		\begin{equation}
				\lim_{n\to\infty}\sun{B}_{\theta} n^{\theta/(\theta+1)}\int_0^\infty \exp\left(au - \frac{u^{(\theta+1)}}{(\theta+1)!}\right)du = 1.
		\label{eqn:gt1}
		\end{equation}
\label{thm:th_gt1}
\end{theorem}
\begin{proof}
		The proof follows a similar reasoning as in the proof of Theorem $14$ in \cite{Jagerman74}. From \eqref{eqn:bt_int} and using Lemma \ref{lem:asym_h}
		\begin{align}
			1 &= \lim_{n\to\infty} \sun{B}_{\theta} n\rho \int_0^\infty \left(\sum_{k=0}^{\theta}\frac{1}{k!} t^k \right)^{n}e^{-t n\rho} dt \\
			&\sim \lim_{n\to\infty} \sun{B}_{\theta} n\rho \int_0^\infty \exp\left(nt - n\frac{t^{(\theta+1)}}{(\theta + 1)!} -t n\rho\right) dt\\
			&\sim \lim_{n\to\infty} \sun{B}_{\theta}  n\rho\int_0^\infty \exp\left(an^{1/(\theta+1)}t - n\frac{t^{(\theta+1)}}{(\theta + 1)!}\right) dt
		\end{align}
		Setting $u = tn^{1/(\theta+1)}$ in the integral gives:
		\begin{equation}
			\sim \lim_{n\to\infty} \sun{B}_{\theta} (n^{\theta/(\theta+1)}+a) \int_0^\infty \exp\left(au - \frac{u^{(\theta+1)}}{(\theta + 1)!}\right) du
			\label{eqn:bt_u}
		\end{equation} 
\end{proof}
Note that $a$ can be positive or negative, which means that even with a total charge larger than the number of servers, the blocking probability can decay  to $0$ provided that \eqref{eqn:hw_th} is satisfied asymptotically. When $a=0$ and using simple computations, the Theorem leads to the following corrollary:

\begin{corollary}
If $\rho=1$: 
\begin{equation}
\sun{B}_\theta \sim {(\theta +1)!^{1 \over \theta+1} \over \theta+1 } \Gamma\Big({1 \over \theta+1}\Big) n^{-\theta/(\theta+1)},
\end{equation}
where $\Gamma$ is the Gamma function.
\end{corollary}

Note that for $\theta=1$, we retrieve that:
			\begin{equation}
					\sun{B}_1 \sim (0.5\pi n)^{-1/2}.
			\end{equation}

\subsection{Moderate deviations}
\label{ssec:md}
Using the previous estimates on the blocking probability (i.e. on the normalizing constant of the stationary distribution)
we can now characterize the deviations around $\hp$ of size smaller than $O(n)$ for a fixed value of $\rho$. Three amplitudes of deviations 
will be identified according to whether $\rho < 1$, $\rho = 1$, 
and $\rho > 1$. The proof for the three results in this subsection appear in the appendix and supplementary material.

The first result is for $\rho < 1$ and is a central-limit-theorem-type scaling when the deviations around the mean are of the order of $\sqrt{n}$.


\begin{theorem}
\label{thm:theta_gt_1}
	 For $\rho < 1$, 
			\begin{align}
				\frac{1}{\sqrt{n}}\left({\big(S^{(n)}(\infty)\big)_{0\le i <\theta} - n(\hat{p})_{0\le i <\theta}}\right) \xrightarrow[n\to\infty]{d} \mathcal{N}(0,\Sigma),
			\end{align}
			where
			\begin{equation}
			\begin{split}
			\Sigma^{-1} &= \psi(1,1,\hdots,1)\cdot(1,1,\hdots,1)^\top \\
				&-\left(\frac{1}{\theta - \hat{c}}\right)(\theta,\theta-1,\hdots,1)\cdot(\theta,\theta-1,\hdots,1)^\top \\
				&+ \left(\begin{array}{cccc}
					1/\hat{p}_0 & 0 & \hdots & 0 \\
					0 & 1/\hat{p}_1 & \hdots & 0 \\
					\vdots & \hdots & \ddots & \vdots \\
					0 & 0 & \hdots & 1/\hat{p}_{\theta-1} 
				\end{array}\right)
			\end{split}
			\end{equation}
\end{theorem}
\begin{corollary}
For $\rho < 1$, $\theta=1$, we have $\hc = \rho$, $\hp_0 = 1-\rho$ and $\psi = \rho^{-1}$ leading to $\Sigma^{-1} = \rho^{-1}$ and:
	\begin{equation}
					\frac{\sun{S}_0(\infty) - n(1-\rho)}{\sqrt{n\rho}} \xrightarrow[n\to\infty]{d} \mathcal{N}(0,1).
	\end{equation}
\end{corollary}

The next case corresponds to $\rho =1$. For $a,z \in \R$ and $\theta \geq 1$, define
\begin{equation}
  \hPhi(z;a) = \int_{z}^\infty \exp\left(au - \frac{u^{(\theta+1)}}{(\theta+1)!}\right) du.
 \label{eqn:wphi}
  \end{equation}
For $\theta=1$ and $a=0$,  $(2\pi)^{-1/2}\hPhi$ reduces to the complementary cumulative distribution function of the standard normal distribution.
\begin{theorem}
\label{thm:theta_gt_2}
			For $\rho = 1$ and $z \in \R_+$, 
			\begin{equation}
					\lim_{n\to\infty}	\P\left(\frac{S^{(n)}_{\theta-1}(\infty)}{n^{\theta/(\theta+1)}} > z\right) = \frac{\hPhi(z;0)}{\hPhi(0;0)},
			\end{equation}
		 	
\end{theorem}
\begin{corollary}
For $\rho = 1$, $\theta = 1$, and $z > 0$,
			\begin{equation}
					\lim_{n\to\infty}	\P\left(\frac{S^{(n)}(\infty)}{\sqrt{n}} > z\right) = 2(1 - \Phi(z)),
			\end{equation}
			where $\Phi$ is the distribution function of the standard normal distribution.
\end{corollary}
Unlike in the $\rho < 1$ case, the deviations are now no longer of $O(\sqrt{n})$ but are of higher order. On the other hand,  
the fluctuations take $\sun{S}$ with high probability only to states with either $\theta - 1$ or $\theta$ jobs. All other configurations are on a scale lower 
than $n^{\theta/(\theta+1)}$. This is in contrast with the behavior for $\rho < 1$ where the fluctuations can take the process to states with number of jobs 
ranging from $0$ to $\theta$. Thus, for $\rho=1$, conditioned on being accepted, a customer has a high probability of being routed to a server $\theta-1$ jobs.
This property has a direct consequence on the state information the dispatcher needs to take routing decisions. We shall elaborate upon this in Section \ref{sec:eng}. 

Finally, for $\rho > 1$, the following result shows that the deviations around $n\theta$ are of $O(1)$ and are geometrically distributed. Moreover, the
excursions take $\sun{S}$ only to states with $\sun{S}_{\theta-1} > 0$ and $\sun{S}_i = 0$ for $i < \theta-1$.  That is, at a random time, there will 
be a geometrically distributed number of servers with $\theta-1$ clients and there will be no servers with less than $\theta-1$ clients.
We give more precise on the blocking probability later on.
\begin{theorem}
\label{thm:theta_gt_3}
		\item For $\rho > 1$, 
			\begin{equation}
					S^{(n)}_{\theta-1}(\infty) \xrightarrow[n\to\infty]{d}  Geo(\rho^{-1}),
			\end{equation}
			and the blocking probability is
			\begin{equation}
					\sun{B}_\theta \sim 1-\rho^{-1}.
			\end{equation}
\end{theorem}

The previous theorems give a more precise characterization
of the system state as well as its performance for a fixed value of $\rho$ (not depending on $n$)
both in terms of blocking and waiting time. In particular, for $\rho <1$, the blocking is exponentially small in $n$
while the mean sojourn time is 
$$ \frac{1}{\rho} \sum_{i} {\theta-i \over \theta -\hat c} i\hat p_i,$$ 
with a deviation of order ${1 \over \sqrt{n}}$.

%

%% file: 06_numerical.tex
\subsection{Numerical experiments}
\label{sec:numexp}

We first provide a comparison in Figure \ref{fig:comp}, of the blocking probability obtained by the insensitive policy analysed in this paper with that of two other policies, namely JSQ and JIQ. 
 The results were obtained through simulations. There are $20$ servers each with a buffer size of $10$. Two different job-size distributions were used: {\it (i)} exponential; and
{\it (ii)} a discrete distribution, which we call custom, with point masses at $0.1$ and $10$. The probability of job-size being $0.1$ (resp. $10$) was  $9/9.9$ (resp. $0.9/9.9$).  
 
 The JSQ policy is known to be optimal for exponential job-size distributions and homogeneous server speeds, and thus gives a natural benchmark for comparison. The JIQ policy is an interesting policy from the practical point of view as it requires little state information compared to JSQ and the insensitive policy while at the same time it is optimal in the mean-field limit, that is its  blocking probability goes to $0$ when the number of servers goes to $\infty$. We have not included the JSQ(d) policy in our comparison
 as this policy is not optimal in the mean-field limit.
We observed in the simulations that in the symmetric case and for high loads, the performance of JSQ and JIQ changes very little when changing the
job-size distribution. 
\begin{figure}
		\centering
		\captionsetup{width=0.8\textwidth}
		\includegraphics[clip, trim=5cm 15cm 3.5cm 4.5cm,width=0.8\textwidth]{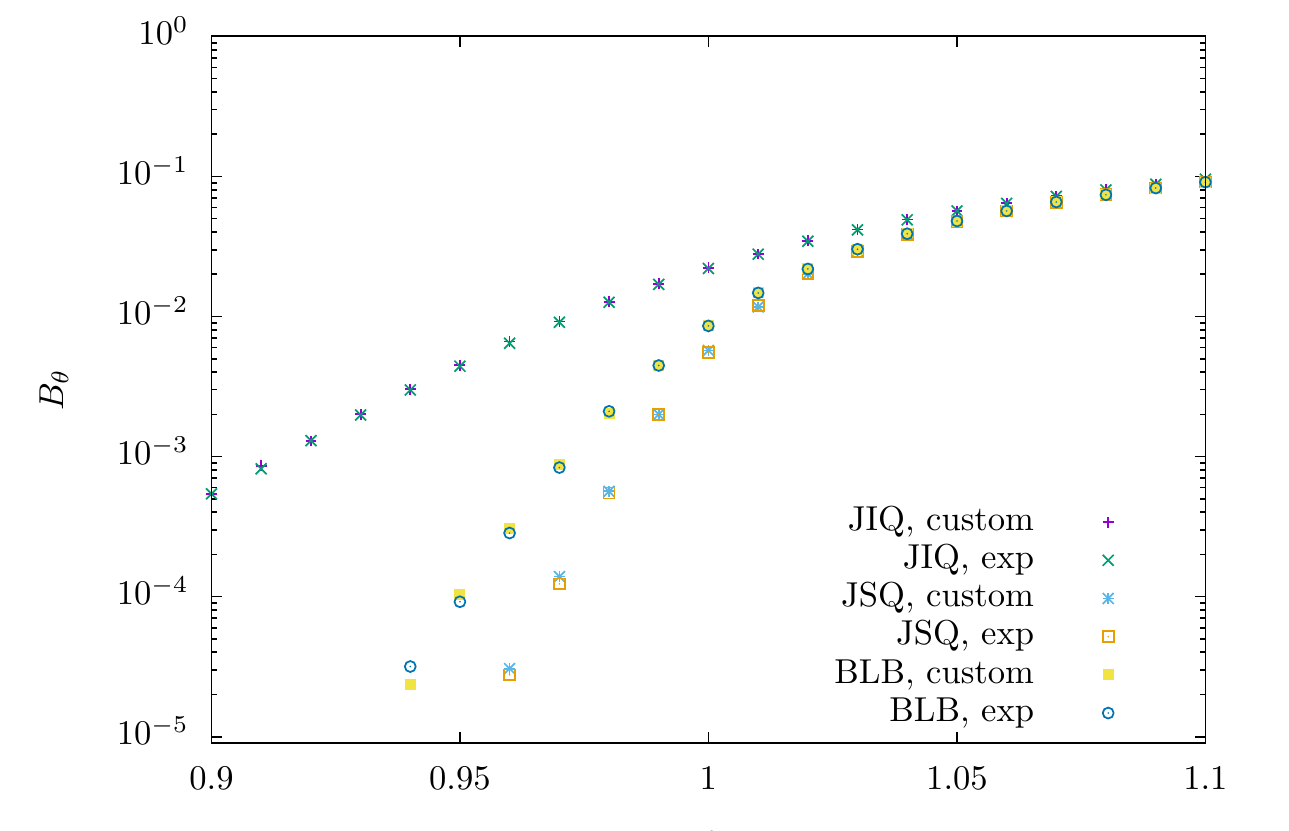}
		\caption{Comparison of the blocking probability for different load balancing policies.  Number of servers is $20$. Buffer size is $10$.}	
	\label{fig:comp}
\end{figure}

While JIQ requires less state information, it can be see that even for a load of $0.9$, a few orders of magnitude of gains can be obtained by using the state information. The drawback of JIQ comes from the fact that at high loads, it behaves more and more like Bernoulli routing since there are fewer empty servers available. Thus, while
JIQ is optimal in the mean-field limit, the number of servers required to get close this limit will be much higher than that of JSQ or the insensitive policy, motivating the asymptotics expressions for fixed $n$ that we provided. For systems with a smaller number of servers in which state information can be obtained relatively cheaply, JSQ or balanced policies can give a considerable performance advantage over JIQ. 

As mentioned in the Introduction, in the case of symmetric speeds, our motivation for studying the insensitive policy comes mainly from the fact that precise asymptotic estimates can be obtained which is obviously not the case with JSQ and JIQ. For asymmetric speed,
insensitive policies might actually present performance gains over JSQ but this falls out of the scope of this paper.

We now illustrate the relationship between the blocking probability and the various parameter such as $\theta$,
$n$ and $\lambda$ for the insensitive policy. First, we evaluate the predictive abilities of some of the results obtained in this section by comparing them 
with the blocking probability obtained from simulating the Markov chain $\sun{S}$. In Figure \ref{fig:n200}, 
we plot the blocking probability for  $n=200$ servers and for different values of $\rho$ and $\theta$.
The theoretical values were calculated using Theorem \ref{thm:btheta} for $\rho < 1$,  Theorem \ref{thm:theta_gt_2} for $\rho=1$,
 and Theorem \ref{thm:theta_gt_3} for $\rho > 1$.

We observe that even for $n=200$ the prediction is already reasonably accurate for $\theta = 2$ and $\theta=3$, 
except for loads very close to
1 where the accuracy is less (this comes from a singurality in the expression of the blocking probability at 1).

\begin{figure}
		\centering
		\captionsetup{width=0.8\textwidth}
		\includegraphics[clip, trim=5cm 15cm 3.5cm 4.5cm, width=0.8\textwidth]{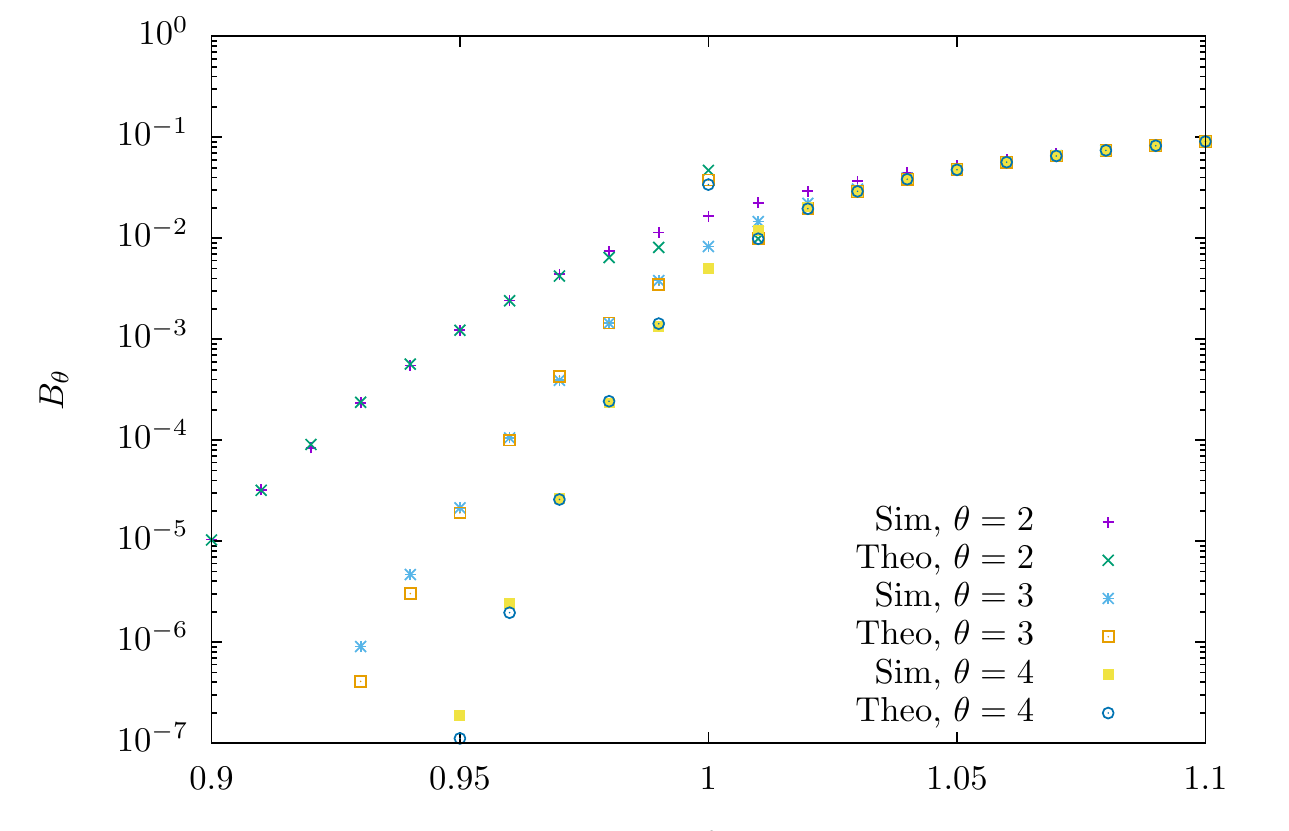}
		\caption{Comparison of the blocking probability computed from Theorems \ref{thm:btheta} and \ref{thm:theta_gt_3} with that obtained from simulations. 
					Number of servers is $200$.}	
	\label{fig:n200}
\end{figure}

Theorem \ref{thm:th_gt1} says that the decay of the blocking probability goes from polynomial to exponential when the load per server
$\rho$ is below $1 + an^{-1/\theta}$. As $\theta$ increases the frontier between the exponential and polynomial decay goes closer to $\rho=1$.
In other words, for a given $n$ as $\theta$ increases, the blocking probability starts to decay exponentially from a value to $\rho$ which is
closer to $1$. This phenomenon is shown in figure \ref{fig:ld}, in which the blocking probability was computed using Theorem \ref{thm:btheta}. 
\begin{figure}[h]
		\centering
		\captionsetup{width=0.8\textwidth}
		\includegraphics[width=0.8\textwidth]{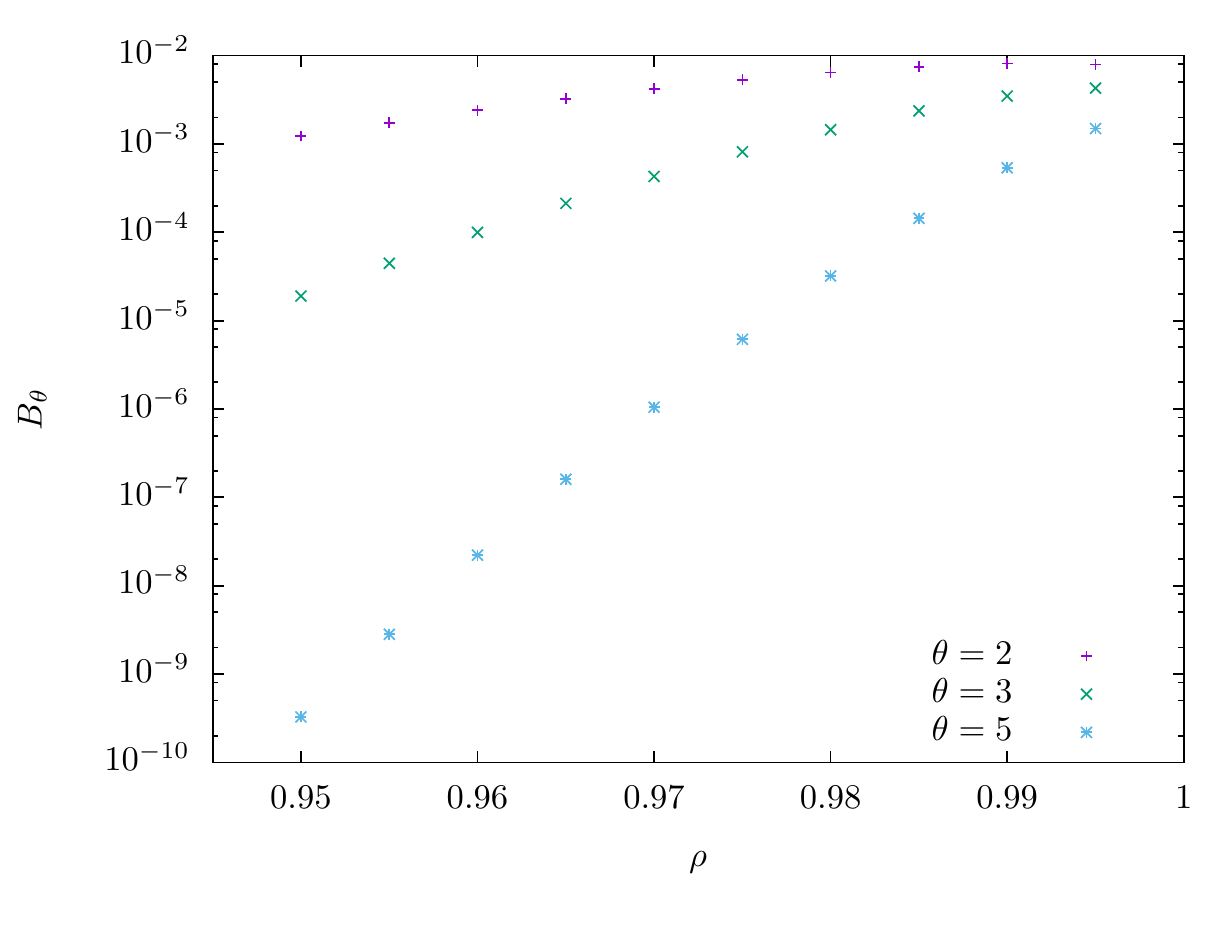}
		\caption{Decay of the blocking probability as predicted by theorem \ref{thm:btheta}. Number of servers is $100$.}	
	\label{fig:ld}
\end{figure}

For the final comparison, we illustrate the benefits of resource pooling. We shall compare three different systems, which we shall index by $\theta$, 
each corresponding to $\theta=1,2,3$. In the system $\theta$, there will be $n/\theta$ servers each of service rate $\theta$ and buffer size
of $\theta$. That is, the three systems have the same total service rate but differ in the buffer size. For the $\theta=1$ system we took $300$ servers
so that the $\theta=2$ (resp., $\theta = 3$) system had $150$ (resp., $100$) each of service rate $2$ (resp., $3$). Figure \ref{fig:rp} illustrates
the blocking probabiilty as a function of $\rho$ for the three systems. The data for this plot was obtained through simulations.
\begin{figure}
		\centering
		\captionsetup{width=0.8\textwidth}
		\includegraphics[clip, trim=5cm 15cm 3.5cm 4.5cm, width=0.8\textwidth]{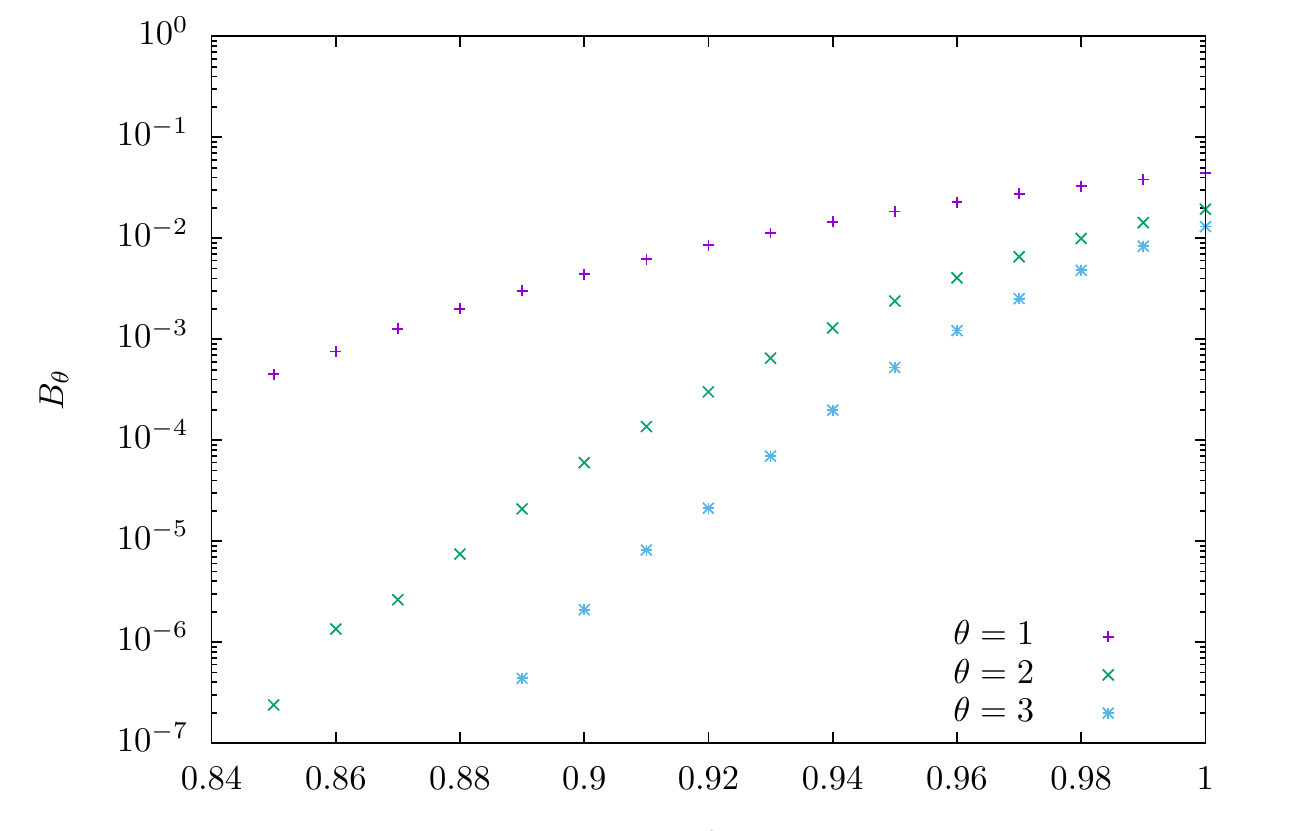}
		\caption{Benefits of resource pooling in terms of the blocking probability. Number of servers in the baseline system is $300$.}	
	\label{fig:rp}
\end{figure}
For example, for a load of $0.9$ the blocking probability for $\theta=1$ system is $4.4\times10^{-3}$ while for the system with $\theta=2$  
the corresponding value is $6\times10^{-5}$ which represents a considerable reduction in the blocking probability.

%% file: 07_practical.tex
\section{Engineering insights and future work}
\label{sec:eng}
\subsubsection*{Performance planning}
The asymptotic analysis of insensitive load balancing allows to give a conservative planning
for managing the performance relationship engaged between delay guarantees depending on $\theta$,
blocking guarantees depending both on $n, \theta$ and $\rho$, and given levels of loads.
Indeed, in many applications, a given level of quality of service in terms of delay has to be reached
and this can be done by fixing $\theta$. For a given buffer depth $\theta$ the mean delay of a job entering the system will be
less than $\theta$ (the server speed and the mean job-size are fixed to 1).
On the other hand, for a given $\theta$ and $n$, we have precisely characterized the asymptotics of the blocking
probability, unveiling the critical load $\rho_c(n)$ as the frontier of the acceptable blocking
probability for most applications.
Hence, one can adapt the number of servers $n$ to cope with a target blocking probability given the load or adapt
the load given the number of servers. Note that this planning is completely out of reach for
specific sensitive policies.

Another way of looking at it is by considering the staffing rule which is the number of servers necessary 
to obtain a vanishing blocking probability in the limit when the total charge is large. In \cite{Jagerman74} and \cite {Whitt84}, the 
staffing rule for $\theta = 1$ was shown to be $\lambda + O(\lambda^{1/2})$, that is, at least these many servers are required to
get a vanishing blocking probability when $\lambda$ is large. 

Theorem  \ref{thm:th_gt1} generalizes the known results for $\theta=1$ to larger values of $\theta$ leading to the following staffing rule:
\begin{proposition}
For a fixed target blocking probability, the number of servers has to scale as
$\lambda + a\lambda^{1/(\theta+1)}$,
where $a$ is determined by the target blocking probability and can computed using (\ref{eqn:gt1}).
\end{proposition}

\subsubsection*{Practical schemes under the critical load}
One of the major criticisms of state-dependent policies  such as JSQ or the policy under study in this paper is that 
the dispatcher needs to know the state of every server in order to route an incoming job. The process of collecting state information 
can add significant delays and lead to lost revenue \cite{Lu11}. Practical policies such as the JSQ(d) \cite{mitzenmacher} or the JIQ \cite{Lu11} play on the trade-off between information
and optimality, and aim to perform much better than state-independent policies while at the same time needing much less information than the 
whole set of servers. For example, JSQ(d), with the knowledge of the state of only $d$ (which can be fixed number independent of $n$) 
servers, has a considerable gain at least in the case of exponentially distributed job-sizes and in the absence of blocking when $d=2$ compared to $d=1$.

While at first glance, the insensitive load balancing policy seems to require full state information, Theorem \ref{thm:theta_gt_2} lends a helping 
hand in alleviating this need. Recall that this theorem has the following implication: for $\rho=1$ and $n$ large, most of the servers will have either 
$\theta$ or $\theta-1$ jobs. One possible scheme to exploit this property is based on the idea first proposed for JIQ, which was motivated by the 
observation that collecting state information at arrival instants should be avoided in order to reduce delays for jobs. In JIQ, the servers inform the 
dispatcher (or leave information on a bulletin board) when they become idle. The dispatcher\footnote{We are assuming a single dispatcher.} then knows 
which servers are idle, and it routes an incoming packet to one of these servers, if there is one, otherwise it routes based on no information. Thus, upon 
arrival a job can be routed immediately based on state information collected previously. 

For the insensitive policy one can conceive a scheme in which servers inform the dispatcher whether they have $\theta-1$ or fewer than $\theta-1$ jobs (this
scheme automatically implies that the dispatcher also knows which servers have $\theta$ jobs). When a job arrives, the dispatcher will need to determine the 
state of only those servers with less than $\theta-1$ jobs. Since this number is expected to be on a smaller scale than $n^{\theta/(\theta+1)}$ 
(thanks to Theorem \ref{thm:theta_gt_2}), one can expect to reduce the information flow between the servers and dispatchers at arrival instants. 
One of our future works will be to characterize precisely the variations in the number of servers with fewer than $\theta-1$ jobs. A back of the envelope
calculation based upon the proof of Theorem \ref{thm:theta_gt_2} leads one to believe that there will $O(n^{(k+1)/(\theta+1)})$ servers with $k$ jobs and
hence $O(n^{(\theta-1)/(\theta+1)})$ servers with less that $\theta-1$ jobs but this remains to be rigorously investigated.

Of course this reasoning is valid for a given blocking probability of  order $ n^{-{\theta \over \theta+1}}$ and 
this could be significantly reduced for other blocking targets (and hence other loads).

\subsubsection*{Multi-speed servers}

This planning is of course simplified by the fact that we considered a symmetric system
depending only on three possibly inter-dependent parameters ($n,\rho, \theta$).
In a future work, we aim at generalizing the analysis to servers with different speeds or even to servers with
state dependent speed.
Though this generalization falls out of the scope of this paper, let us underline the possibility of this analysis
by giving its first step, the expression of the stationary measure for the occupation
of a multi-speed server farm.

Consider a server farm with $n$ servers that are classified according to their speed into $J$ different types. A server of type $j$ has 
speed $c_j$, buffer size of $\theta_j$, and there are $n_j$ servers of type $j$. 

As for the symmetric system, it is convenient here to study the number of servers processing jobs instead of the number of jobs being processed. 
 Let $\mc{S}_j = \{s \in \{0,1,\hdots, n_j\}^{\theta_j} : \sum_{i=0}^{\theta_j-1} s_i \leq n_j\}$, and let $\mc{S}^{(n)}=\prod_{j=1}^J \mc{S}_j$. Further, let 
$\{S^{(n)}(t) \in \mathcal{S}\}_{t\geq 0}$, where $n = \sum_{i=1}^J n_j$.  Let $\sun{S}(t)$ be a random process defined on $\sun{\mc{S}}(t)$, where
the component $(i,j)$ of  $\sun{\mc{S}}(t)$ denotes the number of servers of type $j$ with $i$ customers at time $t$.  

We shall use a boldface font to denote an element of $\mc{S}_j$, and use calligraphic font to denote an element of $\sun{\mc{S}}$.  
So, $\mb{s}_j$ would be an vector in $\mc{S}_j$, and an element $\mathsf{s}\in\sun{\mc{S}}$ can be written as $\mathsf{s} = (\mb{s}_1,\mb{s}_2,\hdots,\mb{s}_J)$.

In state $\ms{s}$, the arrival rate to servers of type $j$ and $i$ tasks is given by
\begin{equation}
	\lambda_{i,j}(\ms{s}) = \frac{(\theta_j - i)s_{i,j}}{\sum_j (n_j\theta _j-\bar{s}_j)}n\rho,
\label{eqn:arr_he}
\end{equation}
where $\bar{s}_j = \sum_i is_{i,j}$ is the total number of tasks in severs of type $j$.

\begin{theorem}
If the job-size distribution is exponential, the process $\mc{S}(t)$ is a reversible Markov process and its stationary distribution of $\sun{\mc{S}}(t)$ is given by
	\begin{align}
		\pi(\mathsf{s}) &= \pi(\ms{0})\frac{(n\theta - \bar{\ms{s}})!}{(n\theta)!}\prod_{j=1}^J \binom{n_j}{\mb{s}_j}\prod_{k=0}^{\theta}\left(\frac{\theta_j!}{(\theta_j-k)!}(n\rho_j)^k\right)^{s_{k,j}}, 
		\label{eqn:pis_he}
\end{align}
where $\bar{\ms{s}} = \sum_j \bar{s}_j$ is the total number of tasks in the system, and $\rho_j = \rho c_j^{-1}$.
\end{theorem}

%
%
%
%
%

Given this first result established, all the steps presented in the presented analysis might (and should) be considered.
This would in particular allow to characterize the optimal trunk reservation parameters $(\theta_i)$ for various trade-off of loads, delays and blocking.

\subsubsection*{Future research directions}

Other than the directions described in the previous subsection, several open questions deserve attention. A natural related model would be the
generalization of the Erlang C model, that is, the model in studied in this work but with a common waiting room where arrivals wait when all the 
servers are in the blocking phase.  More fundamental questions that merit investigation are: 
\begin{itemize}
\item
Can similar results be established for sensitive policies (like join the shortest of d queues among n)?
Are the meaningful scaling similar?
\item
 Can we quantify the optimality gaps for specific families of jobs-size distributions?
 \item
 Can we obtain even finer estimates for the blocking probabilities in the QED regime, in the spirit of the body of work
 establishing precise asymptotics for Erlang's formula \cite{janssen2008}.
\end{itemize}

%% file: appendix.tex
\section{Proof of Theorem 5}
\label{sec:app_fin}

\label{ssec:theta_gt_1}
\begin{proof}

We first prove a local convergence.
Let $q = \hp + \beta/\sqrt{n}$, and let $c = \sum_k k q_k$, $\bar{\beta} = \sum_i i\beta_i$. Since $\sum_k q_k = \sum_k \hp_k = 1$, we have
\begin{align}\label{eqn:sum_beta}
		\sum_k\beta_k = 0,   \ \ \ c = \hat{c} + \frac{\bar{\beta}}{\sqrt{n}}.
\end{align}

We remind the reader that in order to simplify notation, we shall use $p$ instead of $p(c)$. Starting from \eqref{eqn:asym_piq}, 
\begin{align}
		\frac{\pi(q)}{\pi(\hp)} 	&\sim \left(\frac{\psi(c)}{\psi(\hc)}\right)^n e^{\sqrt{n}\bar{\beta}}\prod_k\left(\frac{p_k}{q_k}\right)^{nq_k} \\
					&= \left(\frac{\psi(c)}{\psi(\hc)}\right)^n e^{\sqrt{n}\bar{\beta}}\prod_k\left(\frac{\hp_k}{q_k} \frac{p_k}{\hp_k}\right)^{nq_k} \\
					&= e^{\sqrt{n}\bar{\beta}}\prod_k\left(\frac{\hp_k}{q_k} \frac{p_k\psi(c)^{-1}}{\hp_k\psi({\hc})^{-1}}\right)^{nq_k} \\
					&= e^{\sqrt{n}\bar{\beta}}\prod_k\left(\frac{\hp_k}{q_k}\right)^{nq_k} \prod_k\left(\frac{p_k\psi(c)^{-1}}{\hp_k\psi({\hc})^{-1}}\right)^{nq_k}
\label{eqn:piqp_1}
\end{align}
We shall compute the asymptotics of the two products separately. The first product gives
\begin{align}
		\prod_k\left(\frac{\hp_k}{q_k}\right)^{nq_k}	&= \prod_k\left(1 + \frac{\beta_k}{\hp_k\sqrt{n}}\right)^{-n\hp_k - \sqrt{n}\beta_k}\\
					 &\sim \prod_k \exp\left(- \sqrt{n}\beta_k -  \frac{\beta_k^2}{2\hp_k}\right) \\
					 &= \prod_k \exp\left(-\frac{\beta_k^2}{2\hp_k}\right),
\end{align}
where the last equality follows from \eqref{eqn:sum_beta}. For the second product, from \eqref{eqn:pkc}, 
\begin{align}
		&\log\left(\frac{p_k\psi(c)}{\hp_k\psi({\hc})}\right) \nonumber\\
	&= \log\left(\left( \frac{\theta-\hc-\bar{\beta}/\sqrt{n}}{\rho}\right)^{\theta-k} \left(\frac{\theta-\hc}{\rho}\right)^{-(\theta-k)}\right)\\
		&\sim \log\left(1 - \frac{(\theta-k)\bar{\beta}}{(\theta-\hc)\sqrt{n}}+ \frac{(\theta-k)(\theta -k-1)}{2}\frac{\bar{\beta}^2}{(\theta-\hc)^2 n}\right) \\
		&\sim \frac{-(\theta-k)\bar{\beta}}{(\theta-\hc)\sqrt{n}} - \frac{(\theta-k)}{2}\frac{\bar{\beta}^2}{(\theta-\hc)^2 n}
\end{align}
Thus,
\begin{align}
&\prod_k\left(\frac{p_k\psi(c)^{-1}}{\hp_k\psi({\hc})^{-1}}\right)^{nq_k} \nonumber \\ 
		& \sim \exp\left((n\hp_k + \sqrt{n}\beta_k)\left(\frac{-(\theta-k)\bar{\beta}}{(\theta-\hc)\sqrt{n}} - \frac{(\theta-k)\bar{\beta}^2}{2(\theta-\hc)^2 n}\right)\right) \\
		& \sim -\sqrt{n}\bar{\beta} +\frac{\bar{\beta}^2}{2(\theta-\hc)},
\end{align}
where the equalities \eqref{eqn:sum_beta}, $\sum_k\beta_k = \bar{\beta}$ and $\sum_k\hp_k = \hc$ helped in the simplification.

Substituting the asymptotics of the two products in \eqref{eqn:piqp_1}, we get
\begin{align}
		\frac{\pi(q)}{\pi(p)} &= \exp\left(\frac{\bar{\beta}^2}{2(\theta-\hc)}\right)\prod_k \exp\left(-\frac{\beta_k^2}{2\hp_k}\right)
\end{align}
Consider the exponent on the RHS. Since $\sum_i\beta_i = 0$, we have $\bar{\beta} = \sum_i i \beta_i = \sum_{i=0}^{\theta-1} i \beta_i - \theta\sum_{i=0}^{\theta-1}\beta_i = 
-\sum_i(\theta - i)\beta_i$. Therefore,
\begin{align}
		& \frac{\bar{\beta}^2}{2(\theta-\hc)}  -  \frac{1}{2p_k}\sum_k \beta_k^2  \nonumber \\
		&= \frac{1}{2(\theta-\hc)}\left(\sum_{k=0}^{\theta-1}(\theta - k)\beta_k\right)^2
			 - \frac{1}{2\hp_k}\sum_{k=0}^{\theta-1}\beta_k^2 - \frac{1}{2\hp_\theta}\left(\sum_{i=0}^{\theta-1}\beta_i\right)^2 \\
	 &= \frac{1}{2(\theta-\hc)}\left(\sum_{k=0}^{\theta-1}(\theta - k)\beta_k\right)^2 	
	 	- \frac{1}{2\hp_k}\sum_{k=0}^{\theta-1}\beta_k^2 - \frac{\psi}{2}\left(\sum_{i=0}^{\theta-1}\beta_i\right)^2
\end{align}
Since the multivariate Gaussian distribution has exponent $-\frac{1}{2}\beta\Sigma^{-1}\beta$, we can deduce from the above equation the inverse of the covariance matrix to be one stated in the theorem and the local convergence of $\frac{\pi(q)}{\pi(p)}$ until the Gaussian 
density.

Using the approximation in (\ref{eqn:asym_pi}), combined with the blocking probability estimates obtained in
Theorem \ref{thm:btheta}, it can be easily seen that
$$\pi(\hp) \sim n^{- \theta/2},$$
which in turns implies that for any $q =\hp + \beta/\sqrt{n}$:
$$\pi(q) n^{- \theta/2} \to exp \Big(\frac{1}{2(\theta-\hc)}\left(\sum_{k=0}^{\theta-1}(\theta - k)\beta_k\right)^2 	
	 	- \frac{1}{2\hp_k}\sum_{k=0}^{\theta-1}\beta_k^2 - \frac{\psi}{2}\left(\sum_{i=0}^{\theta-1}\beta_i\right)^2 \Big). $$ 
Generalizing slightly the previous computations, the same would hold for any $q =\hp + (\beta + \epsilon_n)/\sqrt{n} $, with $\epsilon_n$
vanishing when $n$ goes to infinity.
Hence, to derive a global convergence result of the distribution function as stated in the Theorem, we can now appeal to a variant of Scheff\'e's lemma (see for instance Theorem 1.29 in \cite{notesscheffe} with $\delta_i(n)= {1 \over \sqrt{n}}, i =1 \ldots, k $ and $k= \theta$).

\end{proof}

\subsection{Proof of Theorem \ref{thm:theta_gt_2}}
\label{ssec:theta_gt_2}
\begin{proof}
Instead of defining $q$ according to a pre-defined scaling like in the previous proof, we shall this time define it with an arbitrary scaling which shall be made precise later. Let $q = \hp + \sun{\beta}$, where again we have $\sum_k \sun{\beta}_k = 0$.
For $\rho = 1$, $\hp_0 = \hdots = \hp_{\theta-1} = 0$, $\hp_\theta = 1$,  so we shall assume that $\sun{\beta}_k \geq 0$ for $k < \theta$.
Also, for $\rho = 1$, we have $\hc = \theta$ and $\psi(\hc)=1$ so that
\begin{equation}
	c = \theta + \sun{\bar{\beta}}, \; \psi(c) = \sum_{j=0}^\theta \frac{(-\sun{\bar{\beta}})^j}{j!},
\end{equation}
where $\sun{\bar{\beta}} = \sum_k k\sun{\beta}_k <0$. 

Our starting point is again \eqref{eqn:asym_piq} which for the present case reduces to:
\begin{align}
\frac{\pi(q)}{\pi(\hp)} 	&\sim \psi(c)^n e^{n\sun{\bar{\beta}}}\prod_k\left(\frac{p_k}{q_k}\right)^{nq_k}  \label{eqn:rt2}\\
				&=\left(\frac{\psi(c)}{q_\theta}\right)^{nq_\theta}e^{n\sun{\bar{\beta}}}\prod_{k=0}^{\theta-1}\left(\frac{p_k\psi(c)}{\sun{\beta}_k}\right)^{nq_k} \\
				&=\left(\frac{\psi(c)}{q_\theta}\right)^{nq_\theta}e^{n\sun{\bar{\beta}}}\prod_{k=0}^{\theta-1}\left(\frac{1}{(\theta-k)!}\frac{(\theta - \hc - \sun{\bar{\beta}})^{\theta-k}}{\sun{\beta}_k}\right)^{nq_k}\\
				&=\left(\frac{\psi(c)}{q_\theta}\right)^{nq_\theta}e^{n\sun{\bar{\beta}}}\prod_{k=0}^{\theta-1}\left(\frac{1}{(\theta-k)!}\frac{(-\sun{\bar{\beta}})^{\theta-k}}{\sun{\beta}_k}\right)^{nq_k}.
\label{eqn:r1pq}
\end{align}
Since $\sun{\beta} \sim 0$ and  $\sun{\bar{\beta}} < 0$, the value of $k < \theta$ that makes the largest contribution is $\theta-1$. For all other
values of $k$, $\frac{(\sun{\bar{\beta}})^{\theta-k}}{\sun{\beta}_k} \to 0$ with respect to this fraction for $k = \theta - 1$. That is, fluctuations under this scaling
will be visible only in $\sun{S}_{\theta-1}$ and $\sun{S}_\theta$ and not in lower values of $k$, This further implies that, given the number of jobs in the system, there 
is only possible configuration of servers possible. In other words, given the number in the system, we can immediately deduce the configuration:  
$\sun{S}_{\theta-1} = n\theta - nc$ and $\sun{S}_\theta = n-\sun{S}_{\theta-1}$.  Therefore, there in only one vector $p$ in the set $\sun{\mc{P}}_c$. 
As a consequence, the only possible value of $q$ in \eqref{eqn:rt2} is $p$, which then leads to:
\begin{equation}
\frac{\pi(q)}{\pi(\hp)} 	\sim \psi(c)^n e^{n\sun{\bar{\beta}}}.
\label{eqn:pq3}
\end{equation}
 Consider $q_{\theta-1} = \sun{\beta} \geq 0$, where $\sun{\beta}$ is a scalar from now on. 
Since $q_\theta = 1 - \sun{\beta}$, we have $\sun{\bar{\beta}} = -\sun{\beta}$. Let us compute the asymptotics of the term with $\psi$:
\begin{align}
	n\log(\psi(c)) &=  n\log\left(\sum_{j=0}^{\theta}\frac{{\sun{\beta}}^j}{j!}\right) \sim n\left(\sun{\beta} - \frac{{\sun{\beta}}^{\theta+1}}{(\theta+1)!}\right)
\end{align}
where the last asymptotic form is a consequence of Lemma \ref{lem:asym_h}. Substituting the above relation back in \eqref{eqn:pq3}, we get
 \begin{equation}
\frac{\pi(q)}{\pi(\hp)} 	\sim \exp\left(-n\frac{{\sun{\beta}}^{\theta+1}}{(\theta+1)!}\right),
\label{eqn:pq4}
\end{equation}
where we have used the identity $\sun{\bar{\beta}} = -\sun{\beta}$ which was noted previously.

Consequently, the right scaling for $\sun{\beta}$ is $z n^{-1/(\theta+1)}$, for $z > 0$, which means that $\sun{S}_{\theta - 1} = n\sun{\beta}$ lives on a scale of $n^{\theta/(\theta+1)}$. 
As for the proof of the central-limit Theorem, we can pass from local to global convergence combining \eqref{eqn:pq4}, the
estimate on the blocking probabilities given in Theorem \ref{thm:th_gt1}, and
Theorem 1.29 in \cite{notesscheffe} with $\delta_1(n)= {1 \over \sqrt{n}}$ and $k= 1$.

\end{proof}

\subsection{Proof of Theorem \ref{thm:theta_gt_3}}
\label{ssec:theta_gt_3}
\begin{proof}
Following the same steps as in the proof of theorem \ref{thm:theta_gt_2} until \eqref{eqn:r1pq}, we can arrive at the conclusion $\sun{S}_{\theta-1}$ and $\sun{S}_\theta$
will be non-zero. Note that the only difference with the $\rho = 1$ case is that now 
\begin{equation}
	\psi(c) = \sum_{j=0}^\theta \frac{(-\sun{\bar{\beta}})^j}{\rho^j j!},
\end{equation}
and $p_k$ has a factor $\rho^{-(\theta-k)}$.  Going further until \eqref{eqn:pq4} leads us to:
\begin{equation}
\frac{\pi(q)}{\pi(\hp)} 	\sim \exp\left(-n\sun{\beta} + n\frac{\sun{\beta}}{\rho} -n\frac{{\sun{\beta}}^{\theta+1}}{\rho^{\theta+1} (\theta+1)!}\right).
\end{equation}
The only possible scaling, is thus, $\sun{\beta} = zn^{-1}$, which means that the fluctuations of $\sun{S}_\theta$ around $n\theta$ are $O(1)$.

We cannot carry on from this stage onwards in the same line as that in the proof of theorem \ref{thm:theta_gt_3} because to arrive at \eqref{eqn:pq3} we
had assumed that the non-zero fluctuations we increasing with $n$ (this was needed to apply Stirling's approximation). So, we shall work directly with 
the stationary distribution.  From  \eqref{eqn:pis_btho},
\begin{align}
\P(\sun{S}_{\theta-1} = s) &= \sun{B}_\theta s! \frac{n!}{s! (n-s)!}(n\rho)^{-s} \\
					&=  \P(\sun{S}_{\theta-1} = 0) \frac{n!}{(n-s)!}(n\rho)^{-s} \\
					&\sim  \P(\sun{S}_{\theta-1} = 0) \rho^{-s},
\end{align}
which is a consequence of Stirling's approximation.
\end{proof}

\subsection{Concavity of $R$}
\label{ssec:Rconc}

\begin{lemma}
The function $R:\R_+ \to \R$ defined by
		\begin{equation}
				R(t) = \log\left(\sum_{k=0}^{\theta}\frac{t^k}{k!} \right) - \rho t,
		\end{equation}
		is concave.
\label{lem:Rconc}
\end{lemma}
\begin{proof}
		 Recall that
$				g_\theta(t) = \sum_{k=0}^{\theta}\frac{t^k}{k!}.$
		Rewrite $g_\theta(t)$ in terms on the incomplete gamma function using the following steps:
		\begin{align}
				g_\theta(t) &= \frac{1}{\Gamma(\theta+1,0)}\int_{0}^{\infty}\left(t + u\right)^\theta e^{-u}du \\
				&= e^{t}\tilde{\Gamma}(\theta+1, t), \label{eqn:g_incgam}
		\end{align}
		where $\tilde{\Gamma}$ is the normalized incomplete gamma function, that is, $\tilde{\Gamma}(m,x) = \frac{\Gamma(m,x)}{\Gamma(m,0)}$.
		
		To show the concavity of $R$, we shall show that its second derivative is negative. Note that $g^\prime_\theta(t) = g_{\theta-1}(t)$
		so that
		\begin{align}
				R^\prime(t) = \frac{g_{\theta-1}(t)}{g_\theta(t)} - \rho,
				\label{eqn:fprime}
		\end{align}
		and
		\begin{align}
				R^{\prime\prime}(t) &= \frac{g_{\theta}(t)g_{\theta-2}(t) - g_{\theta-1}(t)^2}{g_\theta(t)^2} 
				\label{eqn:f2prime} \\
				&= \frac{\tilde{\Gamma}(\theta+1,t)\tilde{\Gamma}(\theta-1, t) - \tilde{\Gamma}(\theta, t)^2}{\tilde{\Gamma}(\theta+1, t)^2}.
		\end{align}
		It is shown in \cite{AB2012} that $\tilde{\Gamma}$ viewed as a function of $\theta$ is log-concave for all $t > 0$. We can thus infer that $R$ is concave in $(0,\infty)$.
\end{proof}

%% file: supplementary.tex
\subsection{Generating functions}
Let 
\begin{equation}
		\M^{(n)}(\mb{z}) = \sum_{s}\pi(s)\prod_{k=0}^\theta z_k^{s_k}
\end{equation} 
be the moment generating function of $S^{(n)}$.
\begin{theorem}
\begin{equation}
		\M^{(n)}(\mb{z}) = B_\theta (n\rho)\int_0^\infty \left(\sum_{k=0}^{\theta}\frac{1}{k!} t^k z_{\theta-k}\right)^{n}e^{-tn\rho} dt.
\label{eqn:gf1}
\end{equation} 	
\label{thm:gf1}
\end{theorem}
\begin{proof}
From \eqref{eqn:pis_btho} and using the fact that $x! = \int_0^\infty t^x e^{-t} dt$ and $(n\theta - \bar{s}) = \sum_k(\theta-k)s_k$,
\begin{align}
		\bar{\M}^{(n)}(z) &= B_\theta\sum_s \int_0^\infty t^{\sum_k(\theta -k)s_k} e^{-t} dt \binom{n}{s}\prod_{k=0}^{\theta}\left(\frac{(n\rho)^{-(\theta-k)}z_k}{(\theta-k)!}\right)^{s_k} \\
		&=B_\theta\sum_s \binom{n}{s}\int_0^\infty \prod_{k=0}^{\theta}\left(\frac{((n\rho)^{-1}t)^{(\theta-k)}z_k}{(\theta-k)!}\right)^{s_k}e^{-t} dt \\
		&=B_\theta\int_0^\infty \sum_s \binom{n}{s}  \prod_{k=0}^{\theta}\left(\frac{((n\rho)^{-1}t)^{(\theta-k)}z_k}{(\theta-k)!}\right)^{s_k}e^{-t} dt\\
		&=B_\theta\int_0^\infty \left(\sum_{k=0}^{\theta}\frac{1}{k!}t^k(n\rho)^{-k} z_{\theta-k}\right)^{n}e^{-t} dt,\label{eqn:in1}
\end{align}
where the last identity a consequence of the multinomial theorem followed by a relabelling of the index inside the sum. Finally, making the transformation $t \mapsto tn\rho$ inside the 
integral completes the proof.
\end{proof}

Let
\begin{equation}
\bar{\M}^{(n)}(z) = \sum_{j}\left(\sum_{s: \sum_k ks_k = j}\pi(s)\right)z^j
\end{equation} 
be the moment generating function of the number of tasks in steady state. 
\begin{lemma}
\begin{equation}
\bar{\M}^{(n)}(z) = \M^{(n)}(z^0, z^1, \hdots,z^\theta).
\end{equation}
\label{lem:gf}
\end{lemma}
\begin{proof}
From its definition
\begin{align}
\bar{\M}^{(n)}(z) &= \sum_{j}\left(\sum_{s: \sum_k ks_k = j}\pi(s)\right)z^j \\
			&= \sum_{j}\left(\sum_{s: \sum_k ks_k = j}\pi(s)z^{k s_k}\right)\\
			&= \sum_{j}\left(\sum_{s: \sum_k ks_k = j}\pi(s)(z^k)^{s_k}\right) \\
			&= \sum_s\pi(s)(z^k)^{s_k} \\
			&= \M^{(n)}(z^0, z^1, \hdots,z^\theta).
\end{align}
\end{proof}

\begin{theorem}
\begin{equation}
		\bar{\M}^{(n)}(z) = B_\theta (n\rho)\int_0^\infty \left(\sum_{k=0}^{\theta}\frac{1}{k!} t^k z^{\theta-k}\right)^{n}e^{-tn\rho} dt.
\label{eqn:gf}
\end{equation} 	
\label{thm:gf}
\end{theorem}
\begin{proof}
The result is a direct consequence of Theorem \ref{thm:gf1} and Lemma \ref{lem:gf}.
\end{proof}

\section{Miscellaneous results}
\label{sec:misc}
\begin{lemma}
		For $\theta \geq 1$,
		\begin{equation}
				\log\left(\sum_{i=0}^\theta \frac{t^i}{i!}\right) = t - \frac{1}{(\theta+1)!} t^{\theta+1} + o(t^{\theta+1}),\; 	\mbox{as } t\to 0.
		\end{equation}
\label{lem:asym_h}
\end{lemma}
\begin{proof}
Let $h_\theta(t)=\log\left(\sum_{i=0}^\theta \frac{t^i}{i!}\right)$. The proof is based on computing the coefficients in Taylor series expansion of $h$ around $0$, that is,
The first derivative of $h$ is:
	\begin{equation}
			h^{(1)}_{\theta}(t) = 1 - \frac{t^\theta}{\theta!}g_\theta(t)^{-1},
	\end{equation}
	where $g_\theta(t) = \sum_{i=0}^\theta \frac{t^i}{i!}$, which gives the coefficient of $t$ as $1$. 
	
	For $k\leq \theta$, taking the $k$th derivative of $h^{(1)}$ and evaluating it using the product rule for higher order derivatives, 
	we obtain the $(k+1)$th derivative of $h$ as:
	\begin{align}
			h^{(k+1)}_\theta (t) = -\sum_{j=0}^k \binom{k}{j} \frac{t^{\theta-j}}{(\theta-j)!}(g_\theta(t)^{-1})^{(k-j)},
	\end{align}
	where $(g_\theta(t)^{-1})^{(k-j)}$ is the $(k-j)$th derivative of $g_\theta(t)^{-1}$. Assuming that the derivatives of $g_\theta(t)^{-1}$ do not go to $\infty$ at $t=0$ (which can
	be seen to be true), at $t=0$, the only non-zero derivative is obtained for $k=\theta$ and $j=k$. That is,
	\begin{equation}	
			h_\theta^{(k+1)}(0) = 
				\begin{cases}
					0 & 1 \leq k < \theta; \\
					-1 & k = \theta.
				\end{cases}
	\end{equation}
\end{proof}